\documentclass[11pt]{article}
\usepackage[margin=1in]{geometry}
\usepackage{amsmath,amsfonts,amssymb,amsthm}
\usepackage[T1]{fontenc}
\usepackage{mathpazo, eufrak, tgpagella}
\usepackage[ruled,vlined,linesnumbered,commentsnumbered]{algorithm2e}
\SetKwInput{KwIn}{Input}
\SetKwInput{KwOut}{Output}
\SetKwInput{KwGlobal}{Global variable}
\SetKwInput{KwOracle}{Oracle}
\SetKwComment{Comment}{// }{}
\usepackage{graphicx}
\usepackage{enumerate}
\usepackage{enumitem}
\usepackage{bbm}
\usepackage{verbatim}
\usepackage{hyperref,color}
\usepackage[capitalize,nameinlink]{cleveref}
\usepackage[dvipsnames]{xcolor}
\hypersetup{
	colorlinks=true,
	pdfpagemode=UseNone,
	citecolor=OliveGreen,
	linkcolor=NavyBlue,
	urlcolor=Magenta,
	pdfstartview=FitW
}
\usepackage{appendix}
\usepackage{makecell}
\usepackage{tablefootnote}
\crefname{appsec}{Appendix}{Appendices}
\usepackage{tikz}
\usepackage{pgfplots}
\usepackage{xifthen}
\usepackage{subcaption}
\usepackage{hhline}
\usepackage{mathtools}

\theoremstyle{plain}
\newtheorem{theorem}{Theorem}[section]
\newtheorem{proposition}[theorem]{Proposition}
\newtheorem{lemma}[theorem]{Lemma}
\newtheorem{corollary}[theorem]{Corollary}

\theoremstyle{definition}
\newtheorem{definition}[theorem]{Definition}

\newtheorem*{assumption*}{Assumption}

\theoremstyle{remark}
\newtheorem{remark}[theorem]{Remark}

\crefname{lemma}{Lemma}{Lemmas}
\crefname{theorem}{Theorem}{Theorems}
\crefname{definition}{Definition}{Definitions}
\crefname{fact}{Fact}{Facts}
\crefname{claim}{Claim}{Claims}
\crefname{proposition}{Proposition}{Propositions}

\newcommand{\allone}{\mathbf{1}}

\newcommand{\eps}{\varepsilon}

\newcommand{\R}{\mathbb{R}}

\newcommand{\II}{\mathcal{I}}

\newcommand{\supp}{\mathrm{supp}}

\newcommand{\e}{\mathrm{e}}
\renewcommand{\epsilon}{\varepsilon}
\renewcommand{\emptyset}{\varnothing}
\newcommand{\set}[1]{\left\{#1\right\}}
\newcommand{\tuple}[1]{\left(#1\right)} 

\newcommand{\tp}{\tuple}
\newcommand{\ol}{\overline}
\newcommand{\abs}[1]{\left\vert#1\right\vert}

\newcommand{\ftp}[1]{\left\lfloor#1\right\rfloor}

\def\*#1{\boldsymbol{#1}} 
\def\+#1{\mathcal{#1}} 
\def\-#1{\mathrm{#1}} 
\def\=#1{\mathbb{#1}} 
\def\!#1{\mathfrak{#1}} 

\def\oPr{\mathop{\mathrm{Pr}}}
\renewcommand{\Pr}[2][]{ \ifthenelse{\isempty{#1}}
  {\oPr\left[#2\right]}
  {\oPr_{#1}\left[#2\right]} } 

\def\oE{\mathop{\mathbb{E}}}
\newcommand{\E}[2][]{ \ifthenelse{\isempty{#1}}
  {\oE\left[#2\right]}
  {\oE_{#1}\left[#2\right]} }

\def\oVar{\mathrm{Var}}
\newcommand{\Var}[2][]{ \ifthenelse{\isempty{#1}}
  {\oVar\left[#2\right]}
  {\oVar_{#1}\left[#2\right]} }

\def\oEnt{\mathrm{Ent}}
\newcommand{\Ent}[2][]{ \ifthenelse{\isempty{#1}}
  {\oEnt\left[#2\right]}
  {\oEnt_{#1}\left[#2\right]} }

\newcommand{\PhiEnt}[2][]{ \ifthenelse{\isempty{#1}}
  {\oEnt^\phi\left[#2\right]}
  {\oEnt^\phi_{#1}\left[#2\right]} }

\newboolean{anonymous}
\setboolean{anonymous}{false}

\title{Subquadratic Counting via Perfect Marginal Sampling \\
}
\ifthenelse{\boolean{anonymous}}{}{
\author{Xiaoyu Chen \thanks{Massachusetts Institute of Technology, Cambridge, MA, USA. Email: \texttt{xiaoyu@mit.edu}. Supported by the NSF CAREER grant CCF-2443045, and the Reed Fund at MIT.}
\and Zongchen Chen \thanks{Georgia Institute of Technology, Atlanta, GA, USA. Email: \texttt{chenzongchen@gatech.edu}}
\and Kuikui Liu \thanks{Massachusetts Institute of Technology, Cambridge, MA, USA. Email: \texttt{liukui@mit.edu}. Supported by the NSF CAREER grant CCF-2443045, and the Reed Fund at MIT.}
\and Xinyuan Zhang \thanks{Nanjing University, Nanjing, Jiangsu, China. Email: \texttt{zhangxy@smail.nju.edu.cn}}
}
}

\pgfplotsset{compat=1.18} 
\begin{document}

\maketitle
\thispagestyle{empty}
\begin{abstract}
    We study the computational complexity of approximately computing the partition function of a spin system. Techniques based on standard counting-to-sampling reductions yield $\tilde{O}(n^2)$-time algorithms, where $n$ is the size of the input graph. We present new counting algorithms that break the quadratic-time barrier in a wide range of settings. For example, for the hardcore model of $\lambda$-weighted independent sets in graphs of maximum degree $\Delta$, we obtain a $\tilde{O}(n^{2-\delta})$-time approximate counting algorithm, for some constant $\delta > 0$, when the fugacity $\lambda < \frac{1}{\Delta-1}$, improving over the previous regime of $\lambda = o(\Delta^{-3/2})$ \cite{AFFGW25}. Our results apply broadly to many other spin systems, such as the Ising model, hypergraph independent sets, and vertex colorings.
    
    Interestingly, our work reveals a deep connection between \emph{subquadratic} counting and \emph{perfect} marginal sampling. For two-spin systems such as the hardcore and Ising models, we show that the existence of perfect marginal samplers directly yields subquadratic counting algorithms in a \emph{black-box} fashion. For general spin systems, we show that almost all existing perfect marginal samplers can be adapted to produce a sufficiently low-variance marginal estimator in sublinear time, leading to subquadratic counting algorithms.

\end{abstract}

\newpage
\setcounter{page}{1}

\section{Introduction}
In this paper, we study the computational complexity of approximate counting and sampling, two fundamental algorithmic primitives in high-dimensional statistics with wide-ranging applications. Specifically, suppose that we are given query access to a nonnegative weight function $w: \Omega \to \R_{\geq0}$ on a large but finite state space $\Omega$ (e.g., a structured subset of $\{\pm1\}^{n}$ or $\R^{n}$). We are interested in the following two central computational problems:
\begin{itemize}[label={}]
    \item \textbf{Sampling:} Generate a sample $\sigma \in \Omega$ from the target distribution $\mu(\sigma) \coloneqq \frac{w(\sigma)}{Z}$, where ${Z \coloneqq \sum_{\sigma \in \Omega} w(\sigma)}$ is the normalizing constant for $\mu$, often referred to as the \emph{partition function}.
    \item \textbf{Counting:} Compute $Z$.
\end{itemize}
Note that when $w \equiv 1$, the counting problem becomes estimating the size of the state space $\Omega$, and thus appears naturally in many combinatorial and geometric problems. Seminal examples include the volume of a high-dimensional convex body \cite{DFK91}, the number of perfect matchings in a bipartite graph \cite{JSV04}, and the number of spanning trees in a graph \cite{Kirchhoff1958}. Algorithms for sampling and counting serve as building blocks for a broad range of tasks, including parameter estimation and statistical inference.

For sampling, one popular approach is the \emph{Markov chain Monte Carlo (MCMC)} method: One designs a Markov chain with stationary distribution $\mu$, and then simulates a sufficiently long trajectory to obtain an (approximate) sample. MCMC methods have achieved great success across many classical models in theoretical computer science, and are employed heavily in practice. Furthermore, MCMC methods are typically easy to implement, and there is a rich and beautiful theory for studying rates of convergence.

Unlike sampling, however, there is no general family of algorithms that directly solve the counting problem. Perhaps the simplest approach is to reduce it to (approximately) sampling from $\mu$ and other induced distributions \cite{JVV86}. 
To explain this reduction, consider the following ``local'' variants of counting and sampling, where we assume $\Omega = [q]^{n}$ for a finite alphabet $[q] \coloneqq \{1,\dots,q\}$:
\begin{itemize}
    \item \textbf{Marginal Sampling:} Given $i \in [n]$, sample from the marginal distribution $\mu_{i}$ over $[q]$ defined by $\mu_{i}(c) \coloneqq \Pr[\sigma \sim \mu]{\sigma_{i} = c}$.
    \item \textbf{Marginal Inference:} Given $i \in [n]$ and $c \in [q]$, compute the marginal probability $\mu_{i}(c)$.
\end{itemize}
When the alphabet size $q$ is constant and marginal probabilities are bounded away from zero, marginal sampling and inference are computationally equivalent. Furthermore, as long as the distributions of interest are \emph{self-reducible}, i.e., closed under conditioning on the values of any subset of coordinates, global sampling and counting are computationally equivalent to their local versions under polynomial-time reduction.
To see this, observe that sampling $\sigma \sim \mu$ is equivalent to sequentially sampling $\sigma_{i} \sim \mu_{i}(\cdot \mid \sigma_{1},\dots,\sigma_{i-1})$, the marginal of coordinate $i$ conditioned on the values of the first $i-1$ coordinates.
Similarly, for a fixed $\sigma^{*} \in \Omega$, one can estimate $\mu(\sigma^{*})$ via the chain rule $\mu(\sigma^{*}) = \prod_{i=1}^{n} \mu_{i}(\sigma_{i}^{*} \mid \sigma_{1}^{*}, \dots, \sigma_{i-1}^{*})$ by sequentially calling a marginal inference subroutine, and then approximate $Z$ via $Z = \frac{w(\sigma^{*})}{\mu(\sigma^{*})}$. Most common models of interest are self-reducible by definition, including spin systems (and graphical models more generally), determinantal point processes, bases of matroids, etc. Thus, for self-reducible models, the arguments above provide a polynomial-time reduction from counting to sampling, and vice versa.

Such reduction-based counting algorithms have been incredibly fruitful, as they allow one to leverage powerful sampling algorithms such as MCMC and require only mild structural assumptions on the problem input (e.g., self-reducibility).
However, an unfortunate limitation of these reduction-based methods is that their running time is typically at best quadratic in the dimension $n$. 
Specifically, suppose we wish to obtain a $(1\pm \eps)$-multiplicative approximation of $Z$ using the reduction to sampling described in the previous paragraph. By the chain rule, we need to estimate $n$ (conditional) marginal probabilities, each of which requires $\tilde{O}(\frac{n}{\eps^2})$ marginal samples to guarantee sufficiently high accuracy; see \cref{lem:reduction-counting-to-inference} for a formal statement. In particular, even in the ideal scenario where each marginal sample takes $O(1)$-time, the combined running time for estimating $Z$ is $\tilde{O}(\frac{n^2}{\eps^2})$.\footnote{We remark that other reductions between counting and sampling are available as well. In particular, the \emph{simulated annealing} algorithm considers ``temperature-adjusted'' distributions $\mu^{(\beta)}(\sigma) \coloneqq \frac{w(\sigma)^{\beta}}{Z(\beta)}$ instead of conditioning \cite{BSVV08,SVV09,Huber15}. Nonetheless, the quadratic-time barrier is present here as well.} This motivates the following fine-grained question:


\begin{quote}\centering
    \textit{Can we bypass the $\tilde{\Theta}(\frac{n^2}{\eps^2})$ barrier for approximate counting?}
\end{quote}



\subsection{Main Results}
We study subquadratic counting in the context of general \emph{spin systems}, which are important probability distributions originating in statistical physics and widely used in machine learning as \emph{undirected graphical models}. They are parameterized by local interactions between vertices of a host (hyper)graph, with paradigmatic examples including the \emph{Ising model}, the \emph{hardcore model}, and the uniform measure over proper colorings of a graph; the precise definitions of these models can be found in \cref{sec:model}. 

Let us begin with the hardcore model, which we will use as a running example throughout this introduction.
Let $G=(V,E)$ be an undirected graph, and fix a \emph{fugacity} parameter $\lambda \geq 0$.
Denote the collection of independent sets of $G$ by $\II$.
The Gibbs distribution of the hardcore model on $G$ assigns every independent set $\sigma \in \II$ probability mass $\frac{\lambda^{|\sigma|}}{Z}$, where $Z = \sum_{\sigma \in \II} \lambda^{|\sigma|}$ is the partition function, also known as the (univariate) independence polynomial in combinatorics.
For approximate counting, our goal is to estimate $Z$ up to $\eps$-multiplicative error.

The hardcore model originates from statistical physics as an idealized model for a gas. It has attracted significant attention in physics and computer science due to its fascinating phase transition, both statistically and computationally. 
For a given maximum degree $\Delta \geq 3$, let $\lambda_{c}(\Delta) \coloneqq \frac{(\Delta-1)^{\Delta-1}}{(\Delta-2)^{\Delta}} \approx \frac{e}{\Delta-1}$.
This critical threshold characterizes exactly the easiness/hardness of approximate counting and sampling for the hardcore model on the class of graphs of maximum degree $\Delta$: Polynomial-time algorithms exist if and only if $\lambda \le \lambda_c(\Delta)$ (assuming $\mathsf{RP} \neq \mathsf{NP}$) \cite{weitz2006counting,CCYZ25,Sly10,SS12,GSV16}. Furthermore, recent developments on Markov chain mixing times have yielded $\tilde{O}(n)$-time sampling \cite{Liu23,CSV25} and therefore, again using standard counting-to-sampling reductions, $\tilde{O}(\frac{n^2}{\eps^2})$-time approximate counting throughout the subcritical regime $\lambda < \lambda_c(\Delta)$.

We furnish a subquadratic approximate counting algorithm when $\lambda < \frac{1}{\Delta-1}$, a constant factor away from the critical point $\lambda_c(\Delta)$.

\begin{theorem}[Informal version of \Cref{corollary:hardcore}; see also \Cref{thm:2-spin-reduction}]
\label{thm:hc-intro}
Let $\Delta \ge 3$ be an integer and $\lambda < \frac{1}{\Delta-1}$ be a real number. Then, there exists an algorithm that, for any $n$-vertex graph $G$ of maximum degree $\Delta$ and any $\eps>0$, approximates the partition function of the hardcore model on $G$ with fugacity $\lambda$ with $\eps$-multiplicative error. The running time of the algorithm is
$\widetilde{O}((\frac{n}{\eps})^{2-\delta})$ where $\delta=\delta(\Delta,\lambda)>0$ is a constant.
\end{theorem}

Prior to our work, subquadratic counting algorithms for the hardcore model were known only for fugacity $\lambda = o(\Delta^{-3/2})$ \cite{AFFGW25}. Notably, this range of $\lambda$ is far more restrictive than the critical threshold $\lambda_{c}(\Delta) \approx \frac{e}{\Delta-1}$.
The threshold $\frac{1}{\Delta-1}$ in our result is natural, as it corresponds to the Dobrushin uniqueness threshold for Glauber dynamics and also the threshold of the Anand--Jerrum marginal sampler, as we shall elaborate in \cref{subsec:overview}.

Beyond the hardcore model, we vastly expand the range of subquadratic counting in spin systems.

\begin{theorem}[Informal]
    Consider $n$-vertex graphs/hypergraphs of constant maximum degree $\Delta \ge 3$. There exist algorithms that approximate the partition function of the following models in $\widetilde{O}(n^{2-\delta})$ time, where $\delta$ is a constant depending on the maximum degree $\Delta$ and model parameters but independent of $n$.
    \begin{itemize}
        \item (\Cref{cor:polymer}) The hardcore model on bipartite $\rho$-expander graphs when $\lambda \ge (3\Delta)^{10/\rho}$;
        \item (\Cref{cor:coloring}) Proper $q$-colorings when the number of colors $q > 65\Delta$;
        \item (\Cref{cor:spin}) The Ising model with parameter $\abs{\beta - 1} \le \frac{1}{2\Delta}$; 
        \item (\Cref{cor:hyper-indset}) Hypergraph independent sets on $k$-uniform hypergraphs when $2^{k/2} \ge \sqrt{8e} k^2 \Delta$.
    \end{itemize}
\end{theorem}
\begin{remark}
More generally, our results extend to the \emph{abstract polymer models} satisfying certain sampling conditions, and \emph{soft-constraint} spin systems where adjacent vertices have sufficiently weak interactions.
\end{remark}

The core ingredient of our algorithms is the \emph{perfect} marginal sampler, a (typically Las Vegas) algorithm which, upon termination, outputs the spin of a single site drawn \emph{exactly} according to its marginal distribution; no total variation error is permitted. Our main conceptual contribution is that, if we have access to a constant-time (in expectation) perfect marginal sampler, we can speed up the standard reduction from counting to sampling and obtain a subquadratic-time counting algorithm in the context of spin systems: 
\begin{quote}\centering
    \textit{Constant-time \textbf{perfect} marginal sampling 
    $\quad \Longrightarrow \quad$ 
    \textbf{Sub}quadratic counting}
\end{quote}

It is perhaps surprising that such an acceleration can be achieved simply by switching the marginal sampling guarantee from Monte Carlo to Las Vegas.
Moreover, unlike the dichotomy between exact and approximate counting, where the former is $\#\mathsf{P}$-hard even though the latter can admit efficient algorithms, perfect marginal samplers exist in many of the same settings in which approximate marginal samplers exist. For spin systems, morally this is because both often directly leverage some sufficiently strong form of \emph{correlation decay} in the system.

We leave as a tantalizing open question whether or not subquadratic-time counting algorithms exist for the hardcore model on bounded-degree graphs all the way up to the criticality. Towards resolving this question, we give an entirely \emph{black-box} reduction to perfect marginal sampling following the above theme. We refer the readers to \Cref{sec:tasks} for formal definitions of the computational tasks involved. This reduction works in the more general setting of anti-ferromagnetic two-spin systems in the uniqueness regime. 
For a graph $G=(V,E)$ and parameters $(\beta,\gamma,\lambda)$, a two-spin system assigns every configuration $\sigma \in \{0,1\}^V$ a weight
\begin{align*}
    w(\sigma) \coloneqq \beta^{m_1(\sigma)} \gamma^{m_0(\sigma)} \lambda^{n_1(\sigma)},
\end{align*}
where $m_i(\sigma)$ denotes the number of edges with both endpoints assigned $i$ for $i \in \{0,1\}$, and $n_1(\sigma)$ is the number of vertices receiving $1$.
If $\beta = 0$ and $\gamma = 1$, this recovers the hardcore model. 
Although the uniqueness regime is not formally defined this way, from a computational complexity point of view, it characterizes all parameters $(\beta,\gamma,\lambda)$ where polynomial-time approximate counting and sampling algorithms exist.
We refer to \cref{sec:model} for more details.

\begin{theorem} \label{thm:2-spin-reduction}
    Let $\delta \in (0,1)$ be a real and $\Delta \geq 3$ be an integer.
    Suppose there is an $O(1)$-expected-time perfect marginal sampling algorithm for any of the following systems on graphs with maximum degree $\Delta$:
    \begin{itemize}[nosep]
        \item The hardcore model with fugacity $0 < \lambda \leq (1 - 4\delta)\lambda_c(\Delta)$, where $\lambda_c(\Delta) \coloneqq \frac{(\Delta-1)^{\Delta-1}}{(\Delta-2)^\Delta}$;
        \item The anti-ferromagnetic two-spin system where parameters $(\beta,\gamma,\lambda)$ are up-to-$\Delta$ unique with gap $\delta$.
    \end{itemize}
    Then, for the same class of systems on $n$-vertex graphs with maximum degree $\Delta$, there is an $\mathsf{FPRAS}$ for the partition function that runs in expected time $\widetilde{O}(n^{2-x}\epsilon^{-2(1-x)})$ with $x \coloneqq (\log \frac{1}{1-\delta})/ (\log \frac{\Delta}{1-\delta}) \in (0,1)$, where all essential parameters $\Delta, \beta, \gamma, \lambda, \delta$ are assumed to be constant.
\end{theorem}

The proof of \Cref{thm:2-spin-reduction} is provided in \Cref{sec:2-spin-reduction}.

\begin{remark}\label{rmk:marginal-sampler-time}
    In our reduction, the instances that we call the black-box perfect marginal sampler are of size $m \coloneqq O(n^{2 - x})$, where $x = (\log \frac{1}{1-\delta})/ (\log \frac{\Delta}{1-\delta})$.
        Hence, ignoring the $\mathrm{polylog}$ factors, \Cref{thm:2-spin-reduction} runs in expected time $o(n^2)$ as long as the perfect marginal sampler runs in expected time $T(m) = o(n^{x}) = o(m^{x/(2-x)})$ on instances of size $m$ (details are provided in the proof of \Cref{thm:2-spin-reduction}).
        Currently, we are not aware of any perfect marginal samplers that can take this advantage and achieve a better regime.
        We leave this as an open direction to pursue.
\end{remark}

It seems plausible that perfect marginal sampling algorithms exist whenever the spin system satisfies \emph{strong spatial mixing}, a strong form of correlation decay which holds for all bounded-degree graphs whenever the parameters of the spin system lie in the uniqueness regime \cite{li2013correlation}. Such samplers are already known when the input graph is \emph{amenable} \cite{AJ22}.

\subsection{Technical Overview}
\label{subsec:overview}
We contribute two distinct methods for designing and analyzing subquadratic counting algorithms given access to a perfect marginal sampler.

\subsubsection{Faster Variance Decay for an Unbiased Estimator}

As mentioned earlier, one important step in the counting-to-sampling reduction is to compute the marginal probabilities given a marginal sampler. We first present some intuition of the connection between perfect marginal sampling and its inference version.

At first glance, one might think that the inference analog of perfect marginal sampling should be \emph{exact} marginal inference, which requires the exact computation of marginal probabilities without any error.
However, unlike perfect marginal sampling, exact marginal inference is typically $\#\mathsf{P}$-hard since it is equivalent to exact counting.
In fact, the correct analog is \emph{unbiased} marginal inference, which computes a random estimator $\hat{p}$ whose expectation is \emph{exactly} the true marginal $p$. Consider the binary setting for simplicity. Then, a perfect marginal sample valued in $\{0,1\}$ directly serves as an unbiased marginal estimator. We also have a converse: given an unbiased estimator $\hat{p} \in [0,1]$, we can produce a perfect sample by generating from a Bernoulli with mean $\hat{p}$.

Our aim is to obtain an accurate marginal estimator. The performance of an estimator can be measured by its mean square error, which, for an unbiased estimator, is just its variance. Notice that a perfect marginal sample is actually the ``worst'' among all unbiased marginal estimators valued in $[0,1]$, since it has the largest variance.
How can we reduce the variance of our ``worst'' unbiased estimator? The most natural way is to take multiple samples and use the empirical mean, leading eventually to the $\tilde{\Theta}(\frac{n^2}{\eps^2})$-time counting algorithm.

It turns out we can do better in the context of the hardcore model, or more generally, two-spin systems. Think of a toy version where the underlying graph is a complete $d$-ary tree, and we want to estimate the marginal at the root. 
Instead of taking multiple perfect samples at the root, we iteratively draw a perfect sample at each leaf conditioned on the previous samples to form a full configuration $\tau$ on all leaves. We use this boundary condition to produce marginal estimators for vertices one level above the leaves. Then, we use the estimators at height one to produce new ones at height two. We continue propagating until we obtain a marginal estimator at the root, which will be exactly $\hat{p} := \Pr[\sigma \sim \mu]{\sigma_{\mathsf{root}} = 1 \mid \sigma_{\mathsf{leaves}} = \tau}$.
It is straightforward to verify that this marginal estimator is unbiased, i.e., averaging over $\tau$ we get $\Pr[\sigma \sim \mu]{\sigma_{\mathsf{root}} = 1}$.
\footnote{Technically, only the estimator at the root vertex is unbiased, but not the intermediate vertices.}

Along this process, we utilized several important properties of two-spin systems.
\begin{enumerate}
    \item First, there is a way to relate the marginal probability of a vertex with those of its children, in their respective subtrees, which is known as the \emph{tree recursion}.
    \item Second, it suffices to only consider trees due to the beautiful reduction based on self-avoiding walks \cite{weitz2006counting}.
    \item Third, and most importantly, we can show that the variance of the unbiased estimator decays to zero as the height $h$ of the tree increases. This is similar in spirit to a \emph{decay of correlation} phenomenon, a key property that perhaps underlies all existing counting and sampling algorithms in two-spin systems. However, since the running time of this algorithm scales as $d^{h}$, we need the variance to decay by a multiplicative factor strictly less than $\frac{1}{d}$. \emph{We prove that such a fast variance decay holds all the way up to criticality.}
\end{enumerate}
We note that while the basic underlying algorithmic idea was already used in prior work \cite{AFFGW25}, they were limited to the hardcore model in the highly subcritical regime $\lambda = o(\Delta^{-3/2})$. We establish the required variance decay all the way up to uniqueness $\lambda < \lambda_{c}(\Delta) \approx \frac{e}{\Delta-1}$. We also show how to convert the algorithm into a black-box reduction from subquadratic counting to perfect marginal sampling, and hence, extend the applicability of the algorithm to all two-spin systems.

Let us now illustrate our proof approach for \Cref{thm:2-spin-reduction} in more detail, using the hardcore model as an example.
The full proof is given at \Cref{sec:2-spin-reduction}.
The self-avoiding walk tree~\cite{weitz2006counting}, introduced by Weitz, is an important algorithmic tool for deterministic approximate counting in two-spin systems.
In a nutshell, for a given vertex $v \in V$, the self-avoiding walk tree $T = T^v$ is the dictionary tree of all the self-avoiding walks in $G$ started at $v$ (up to some truncation).
Crucially, it holds that $\mu_{G,v} = \mu_{T,v}$.
Hence, it suffices to approximate $\mu_{T,v}(0)$ instead of $\mu_{G,v}(0)$.
The marginal probability of Gibbs distributions on trees enjoys a direct recursion which we now describe.
Let $u_1, \ldots, u_d$ be the children of $v$, and let $T_{i}$ be the sub-tree of $T$ rooted at $u_i$, respectively.
The tree recursion of the hardcore model is given by:
\begin{align*}
    \mu_{T,v}(0) = \frac{1}{1 + \lambda \prod_{i=1}^d \mu_{T_i,u_i}(0)}.
\end{align*}
Note that this does not directly give an efficient algorithm for $\mu_{G,v}$ since $T$ in general is exponentially large.
When $\lambda \leq (1 - \delta)\lambda_c(\Delta)$ where $\lambda_c(\Delta) \approx \frac{\e}{\Delta-1}$, the hardcore model enjoys a \emph{spatial mixing} (or \emph{correlation decay}) property, namely there exists a constant $C > 0$ such that for all trees $T$ of maximum degree $\Delta$ and all integers $\ell \geq 1$,
\begin{align*}
    \sup_{\sigma \in \Omega(\mu_{S(v,\ell)})} \abs{\mu_{T,v}^{\sigma}(0) - \mu_{T,v}(0)} &\leq C \cdot (1 - \delta)^\ell,
\end{align*}
where $S(v,\ell) \coloneqq \set{u \in T \mid \mathrm{dist}_T(v,u) = \ell}$ is the set of vertices $u$ at distance $\ell$ from $v$ in $T$.  

For deterministic counting~\cite{weitz2006counting}, thanks to spatial mixing, we can truncate the recursion to level $\ell = \Theta(\log n)$ while still approximating $\mu_{T,v}(0)$ with enough accuracy.
Since the tree $T$ branches at most $\Delta$, this usually gives a polynomial-time algorithm for estimating $\mu_{T,v}(0) = \mu_{G,v}(0)$ when the maximum degree is constant. 

If randomness is allowed and perfect marginal samplers are available, the previous work of Anand, Feng, Freifeld, Guo, and Wang~\cite{AFFGW25} noticed an improvement by replacing truncation with a random boundary condition drawn from the Gibbs distribution.
In particular, we first generate a perfect sample $X \sim \mu_{T, S(v,\ell)}$ by calling the perfect marginal sampler iteratively. Then, we calculate the marginal probability $P = \mu^X_{T,v}(0)$.
According to the law of total probability, 
\begin{align*}
    \mu_{T,v}(0) = \E[X \sim \mu_{T, S(v,\ell)}]{\mu^X_{T,v}(0)} = \E{P},
\end{align*}
which implies that $P$ is an unbiased estimator of $\mu_{T,v}(0) = \mu_{G,v}(0)$.
To give an $\mathsf{FPRAS}$ for the partition function, according to the standard reduction from counting to sampling (see \Cref{lem:reduction-counting-to-inference}), we only need to make sure that $\Var{P}/\E{P}^2 = O(\Var{P}) = O(1/n)$.  
The advantage here is that, by spatial mixing, the variance decays faster with respect to the level $\ell$:
\begin{align*}
    \Var{P} = \E{(P - \E{P})^2} 
    &= \E[X \sim \mu_{T, S(v,\ell)}]{\tp{\mu^X_{T,v}(0) - \mu_{T,v}(0)}^2} 
    \lesssim (1-\delta)^{2\ell}.
\end{align*}
After a careful analysis, they show that when $\lambda < \frac{1}{\Delta^k(\Delta-1)}$, the spatial mixing decays at the rate $O((1 - \delta)^\ell) = O(\Delta^{-k\ell})$.
This implies that the variance decays at the rate $O(\Delta^{-2k\ell})$ and hence we can choose $\ell \approx \frac{\log n}{2k \log \Delta }$ to make sure that $\Var{P} \approx 1/n$.
Note that after truncation at level $\ell$, the remaining tree is of size at most $\Delta^\ell = O(n^{1/2k})$. 
In particular, when $k>1/2$, we can approximate $\mu_{G,v}(0) = \mu_{T,v}(0)$ with $\lambda = o(\Delta^{-3/2})$ and in time $o(n)$.

This already gives a prototype of the reduction from subquadratic-time counting to perfect marginal sampling for the hardcore model.
However, to get a black-box reduction that works in the entire uniqueness regime, several challenges still remain:
\begin{enumerate}
    \item The bound $\Var{P} \leq C^2(1-\delta)^{2\ell}$ directly obtained from the spatial mixing does not lead to sublinear truncation for larger $\lambda$ in the uniqueness regime. Also, we want a bound that works for all anti-ferromagnetic two-spin systems in the uniqueness regime.
    \item The algorithm in~\cite{AFFGW25} is not a black-box reduction. Though we omit details in this overview, calling perfect marginal samplers on the self-avoiding walk tree in a black-box manner is nontrivial.
    Actually, \cite{AFFGW25} only focuses on the Anand--Jerrum perfect marginal sampler~\cite{AJ22}, and their implementation heavily relies on the fact that the Anand--Jerrum algorithm behaves like a random walk on graphs.
    That is, it won't attempt to visit a vertex far from the one that it is currently focusing on.
\end{enumerate}

We address the first challenge by proving a shaper variance bound that allows sub-linear truncation in the whole uniqueness regime.
For hardcore model with $\lambda \leq (1 - \delta)\lambda_c(\Delta)$, we show
\begin{align} \label{eq:var-bound-hardcore}
    \Var{P} \lesssim \tp{\frac{1-\delta}{\Delta-1}}^\ell,
\end{align}
which significantly improves the previous bound on variance.
With this bound in hand, we can set $\ell \approx \frac{\log n}{\log \frac{\Delta}{1-\delta}}$ to achieve a relatively small variance $\Var{P} \approx 1/n$.
After the truncation at level $\ell$, the size of the remaining tree will be at most $\Delta^\ell \approx n^{\log \Delta/\log\frac{\Delta}{1-\delta}} = o(n)$ for constant $\delta$ and $\Delta$.

Extending the bound in \eqref{eq:var-bound-hardcore} to general anti-ferromagnetic two-spin systems in the uniqueness regime presents a further challenge.
It turns out that for general systems, truncating the tree regularly at level $\ell$ is not the best practice.
Inspired by the notion of \emph{computationally efficient correlation decay}~\cite{li2012approximate, li2013correlation}, we adopt a more sophisticated strategy to truncate the self-avoiding walk tree (see \Cref{algo:boundary-on-tree}), which ends up with similar variance bound as in~\eqref{eq:var-bound-hardcore} for general anti-ferromagnetic two-spin systems (including non-monotone $\gamma > 1$ cases).

We address the second challenge by providing a data structure for the topological (and also colorific) queries on a tree-like graph which simulates the self-avoiding walk tree in terms of the marginal distribution of its root.
Intuitively, for each sub-tree $T_u$ of the self-avoiding walk tree rooted at $u$, there is a sub-graph $G' \ni u$ of the original graph that has not been reached by the self-avoiding walk so far.
In particular, we have $\mu_{G',u} = \mu_{T_u,u}$.
In our implementation, we replace every sub-tree $T_u$ rooted at some vertex $u$ at level $\ell$ of the self-avoiding walk tree $T$ to its corresponding graph $G'$.
This won't change the marginal distribution of the root.
However, this new graph is of size at most $n\cdot\Delta^{\ell} = o(n^2)$ and we provide an efficient data structure to answer topological queries of the perfect marginal samplers in a black-box manner.
The details are provided at \Cref{lem:SAW-tree-with-flower} along with its proof.

\subsubsection{Aggregate Perfect Marginal Sampler}
As we previously described, the standard reduction from counting to sampling requires generating $N = \tilde{O}(n/\eps^2)$ samples $X_{1},\dots,X_{N} \in \{0,1\}$ to estimate a single marginal by taking the empirical mean $\hat{X} = \frac{1}{N} \sum_{i=1}^{N} X_{i}$. Our key insight is that we can modify many existing perfect marginal samplers to directly sample the random variable $\hat{X}$ in $o(N)$-time, without having to generate all $X_{1},\dots,X_{N}$ individually.
To put our idea in one sentence, we do not flip $N$ coins to count the number of heads, but instead sample from a binomial distribution $\mathrm{Bin}(N,\frac{1}{2})$!

We again illustrate our approach using the hardcore model. We begin by first recalling a beautiful algorithm, originally proposed by Anand and Jerrum \cite{AJ22}, which perfectly samples the spin (occupied or unoccupied) of a single vertex in expected $O(1)$ time when $\lambda < \frac{1}{\Delta - 1}$.


\begin{algorithm}[h]
    \KwIn{Graph $G=(V, E)$, vertex $u \in V$ and subset $\Lambda \subseteq V \setminus \{u\}$}
    \KwOut{A marginal sample $\sigma_u \in \{0,1\}$ following the distribution of $\mu_u^\tau$, where $\tau$ pins vertices $v \in \Lambda$ to be unoccupied;}
    \caption{Anand and Jerrum's perfect marginal sampler: AJ($G,u,\Lambda$)}
    \label{alg:AJ-marginal-sampler-new}
    
    \BlankLine
    Sample $r \sim U(0,1)$\;\label{line:rv-new}
    \If{$r < \frac{1}{1+\lambda}$}{
    \Return 0\;}
    \ForEach{$v \in N(u) \setminus \Lambda$}{\label{line:enumerate-new}
        \If{$\text{AJ}(G,v,\Lambda) = 1 $}{
        \Return 0\;} 
        $\Lambda \gets \Lambda \cup \{v\}$\;
    }
    \Return 1\;
\end{algorithm}

The high-level idea behind Anand and Jerrum's algorithm is to decompose the marginal distribution into two regions: a zone of obliviousness and a zone of indecision. 
Specifically, we can decompose the marginal distribution of $\sigma_u$ as follows: 
\begin{align*}
    \Pr[\mu]{\sigma_u = 1} = \Pr[\mu]{\sigma_u = 1 \land \{ \forall v \in N(u), \sigma_v = 0\} } = \Pr[\mu]{\forall v \in N(u), \sigma_v = 0} \cdot \frac{\lambda}{1+\lambda}.
\end{align*}
Thus, we set $\sigma_u = 0$ with probability $\frac{1}{1+\lambda}$ regardless of the spins in the neighborhood of $u$; this is the oblivious zone.
With the remaining probability $\frac{\lambda}{1+\lambda}$, we set $\sigma_u = 1$ with probability $p \coloneqq \Pr[\mu]{\forall v \in N(u), \sigma_v = 0}$; this is the indecisive zone since $p$ is unknown to us and we need to recursively call the procedure to resolve this uncertainty. 
The key observation is that the recursive stack trace can be formulated as a branching process in which the number of children is $0$ with probability $\frac{\lambda}{1+\lambda}$ and (at most) $\Delta$ otherwise. Hence, when $\lambda < \frac{1}{\Delta - 1}$, this branching process is subcritical and the algorithm has $O(1)$-expected running time.

We propose the following \emph{aggregate} version of the perfect marginal sampler, which simulates $N$ independent executions of Anand and Jerrum's original perfect marginal sampler in a single batch, but only tracks the number of samples with outcome $1$. More precisely, we use an auxiliary random variable $N$ to track the number of samples entering each recursive call. 

\begin{algorithm}[H]
    \KwIn{Graph $G = (V, E)$, vertex $u \in V$, subset $\Lambda \subseteq V \setminus \{u\}$, and positive integer $N$}
    \KwOut{The sum of $N$ independent marginal samples $\sigma_u \in \{0,1\}$, following the distribution of $\mu_u^\tau$, where $\tau$ fixes vertices $v \in \Lambda$ to be unoccupied;}
    \caption{``Aggregate'' perfect marginal sampler: Aggregate-AJ($G,u,\Lambda,N$)}
    \label{alg:AJ-marginal-sampler-aggregate}
    
    \BlankLine
    \If{$N = 0$}{
        \Return 0\;
    }
    Sample $X \sim \mathrm{Bin}(N, \frac{\lambda}{1+\lambda})$\;\label{line:rv-aggregate}
    \ForEach{$v \in N(u) \setminus \Lambda$}{
        $Y \gets \text{Aggregate-AJ}(G,v,\Lambda, X)$\;
        $\Lambda \gets \Lambda \cup \{v\}$ and $X \gets X - Y$\;
    }
    \Return $X$\;
\end{algorithm}

The idea behind the aggregate perfect marginal sampler is simple.
Recall that, to determine the zone of obliviousness vs.~indecision in one call of Anand and Jerrum's sampler, we flip a biased coin with head probability $\frac{\lambda}{1+\lambda}$. However, to make the same decisions for $N$ independent calls, we do not have to flip the coin $N$ times, but instead we can directly sample from a binomial distribution $\mathrm{Bin}(N, \frac{\lambda}{1+\lambda})$ to find the number of calls into each zone. Importantly, sampling $X\sim\mathrm{Bin}(N, \frac{\lambda}{1+\lambda})$ takes only $\tilde{O}(1)$ time.

This saving in running time also carries over to subsequent recursive calls.
The key observation is that the total number of recursive calls is $o(N)$ in expectation. This sublinearity becomes intuitive when we consider the idealized scenario where exactly $\frac{\lambda}{1+\lambda}$-fraction of the total calls ``survive'' and are passed to new recursive calls. Therefore, the total number of recursive calls obeys the recurrence $T(m) = 1 + \Delta \cdot T\tp{\frac{\lambda}{1+\lambda} \cdot m}$, which is sublinear whenever the branching process is subcritical, i.e. $\frac{\lambda \Delta}{1+\lambda} < 1$. Combining this with an amortized $O(\log^2 n)$-time oracle for sampling from binomial distributions, we obtain an algorithm for performing marginal inference in $o(N) = o(\frac{n}{\eps^2})$ time.

This approach can be extended to other perfect marginal samplers tailored to various settings, such as the Ising and polymer models (see \Cref{subsec:aggregate-explicit}). 
However, our previous analysis of the ``aggregate'' sampler relies heavily on the underlying branching process. 
This exposes a key challenge for applications lacking such ``local self-recurrence structures''.
Notable examples include perfect marginal samplers for graph colorings~\cite{LWY26} and hypergraph independent sets~\cite{FGWWY23}, whose efficiency depends on ``global properties'' such as witness graphs or potential functions.

To overcome this difficulty, we present a much more general framework for ``aggregate'' perfect marginal samplers.
For simplicity, we assume the marginal sampler has access to randomness only through fair or biased coins. 
Our framework requires only one mild assumption: the stopping time $T$ has a geometric tail, i.e., there exist constants $C \ge 0$ and $\alpha \in (0,1)$ such that for any $t \ge 0$, it holds $\Pr[]{T \ge t} \le C \cdot \alpha^t$.

\begin{theorem}[Informal version of \Cref{thm:batch-automaton}]\label{thm:general-2}
    Provided that the perfect marginal sampler has a stopping time with a geometric tail, we can compute the empirical mean of $N$ samples in $\widetilde{O}(N^{1-\delta})$ time, where $\delta > 0$ is a constant depending on the geometric decay rate.
\end{theorem}

The geometric tail is not a stringent requirement for perfect marginal samplers.
For perfect marginal samplers associated with subcritical branching processes, this condition is naturally satisfied, as the stopping time $T=\min\{t : X_t \le 0\}$ of a supermartingale $(X_t)_{t \ge 0}$ with bounded increments is known to possess a geometric tail.
For samplers in~\cite{LWY26} and~\cite{FGWWY23} relying on ``global structures'', this geometric tail is likewise guaranteed by analogous reasoning or a counting argument.
Hence, we can apply our framework to many perfect marginal samplers in the literature, in particular the ones for proper colorings and hypergraph independent sets.

\begin{remark}
The aggregate perfect marginal samplers for the hardcore and Ising models in~\Cref{subsec:aggregate-explicit} can be alternatively derived from \cref{thm:general-2}. Though, the dependence of $\delta$ on the maximum degree $\Delta$ is better in \Cref{subsec:aggregate-explicit}: For specific applications such as the hardcore model with $\lambda < \frac{1}{\Delta-1}$ or the Ising model with $\abs{\beta - 1} \le \frac{1}{2\Delta}$, it yields $\delta = \Omega(\frac{1}{\log \Delta})$ while (the proof of) \Cref{thm:general-2} yields only $\delta = \Omega(\frac{1}{\Delta})$.
\end{remark}

\subsection{Prior Work}

\paragraph{On Subquadratic Counting} The only prior work on subquadratic counting in spin systems we are aware of is \cite{AFFGW25}, which gave such an algorithm for the hardcore model on graphs of maximum degree $\Delta$ with $\lambda = o(\Delta^{-3/2})$. Beyond the setting of spin systems, we note that subquadratic counting algorithms are also known for spanning trees~\cite{LPY25} and knapsack solutions~\cite{FJ25}. The former is linear algebraic in nature, and relies crucially on the Matrix Tree Theorem \cite{Kirchhoff1958}. The latter is also tailored specifically to the knapsack problem.

\paragraph{On Approximate Counting in Spin Systems} Prior counting algorithms for spin systems can be broadly placed in one of two categories: The first type is randomized, leveraging the counting-to-sampling reduction of Jerrum--Valiant--Vazirani \cite{JVV86}, and using a Markov chain to solve the sampling problem; see also \cite{BSVV08,SVV09,Huber15} for reductions based on the simulated annealing algorithm. The second type, which has witnessed an explosion of interest over the past two decades, directly solves the counting problem \emph{deterministically}. However, those with provable approximation guarantees often incur a huge polynomial running time; see e.g. \cite{weitz2006counting, Bar16, PR17} and references therein.
In practice, message passing algorithms such as belief propagation perform marginal inference well empirically. However, to the best of our knowledge, approximation error bounds are available only for random problem instances, or are absent entirely.


\paragraph{On Perfect Marginal Samplers} 
A few works have investigated perfect marginal sampling in spin systems. 
\cite{AJ22} showed a perfect marginal sampler for the hardcore model, which is the basis of our work.
\cite{LWY26} designed perfect marginal samplers for soft-constraint spin systems with weak interactions and for graph colorings. \cite{FGWWY23} considered deterministic approximate counting for hypergraph independent sets and hypergraph colorings, and their approach also gives perfect marginal samplers for the former.

\section{Preliminaries} \label{sec:prelim}

\subsection{Notations} \label{sec:notations}

\paragraph{Distributions on Product Spaces}
Let $\Sigma$ be a finite alphabet, which we typically take to be $\Sigma = [q] \coloneqq \{1,\dots,q\}$ when $q \geq 3$, or $\Sigma = \{0,1\}$ when considering the Boolean setting. Let $\mu$ be a distribution over $\Sigma^{V}$, where $V$ is a finite set. We call members $\sigma \in \Sigma^{V}$ \emph{configurations}.
\begin{itemize}
    \item We denote the \emph{support} of $\mu$ by $\Omega(\mu) \coloneqq \{\sigma \in \Sigma^V \mid \mu(\sigma) > 0\}$. 
    \item For a subset $\Lambda \subseteq V$, a partial configuration $\tau \in \Sigma^{\Lambda}$ on $\Lambda$ is called a \emph{pinning} (or \emph{boundary condition}) if it is \emph{feasible}, i.e. there exists a completion $\sigma \in \Omega(\mu)$ consistent with $\tau$ on $\Lambda$.
    Let $\tau_1 \in \{0,1\}^{\Lambda_1}$ and $\tau_2 \in \{0,1\}^{\Lambda_2}$ be consistent pinnings, i.e. they agree on $\Lambda_{1} \cap \Lambda_{2}$. We denote the concatenation of these pinnings by $\tau_1 \cup \tau_2 \in \{0,1\}^{\Lambda_1 \cup \Lambda_2}$.
    \item For any pinning $\tau$, the \emph{conditional distribution} $\mu^\tau$ is the distribution over $\Sigma^{V}$ restricted to configurations consistent with $\tau$, with probabilities proportional to $\mu(\sigma)$. 
    \item For a subset $S \subseteq V$, we denote the marginal distribution on $S$ by $\mu_S$. For a singleton $S = \{v\}$, we write $\mu_v$ rather than $\mu_{\{v\}}$ for brevity.
    \item Combining the above, for a pinning $\tau \in \Sigma^{\Lambda}$ and $S \subseteq V \setminus \Lambda$, we write $\mu_{S}^{\tau}$ for the marginal distribution of $S$ conditioned on $\tau$.
\end{itemize}

\paragraph{Multinomial Distributions} Let $M(N; \nu)$ be the $N$-trial \emph{multinomial distribution} defined with respect to a fixed probability measure $\nu$, a sample of which records the number of occurrences of each element of $\supp(\nu)$ in $N$ independent draws from $\nu$. The representation of such a sample depends on the cardinality of the support of $\nu$:

\begin{itemize}
    \item \textbf{Small Support:} When the support of $\nu$ is a fixed finite set, say $\mathrm{supp}(\nu) = [q]$, a sample from $M(N; \nu)$ will be represented as a \textit{frequency vector} $(o_1, o_2, \dots, o_q) \in \mathbb{N}^q$. Here, each $o_i$ denotes the number of occurrences of the $i$-th element among $N$ independent draws from $\nu$, satisfying the constraint $\sum_{i=1}^q o_i = N$.

    \item \textbf{Large or Infinite Support:} When the support of $\nu$ is large or even infinite, a sample from $M(N; \nu)$ is more naturally expressed as a \textit{sparse mapping} or a set of pairs $\{(x, o_x) : o_x \ge 1\}$, where each $x \in \mathrm{supp}(\nu)$ is an observed value and $o_x$ is its corresponding count.
\end{itemize}

\paragraph{Graphs} Let $G=(V,E)$ be an undirected graph and $u \in V$ be a vertex in $G$. Let $N(u)$ to denote the set of vertices adjacent to vertex $u$, and $N^{+}(u) = N(u) \cup \{u\}$ to denote the inclusive neighborhood.

\subsection{Models}\label{sec:model}

\paragraph{Spin Systems}
One of the main classes of models we work with are \emph{spin systems}, a special class of discrete undirected graphical models (also known as \emph{Markov random fields}) originating in statistical physics. They are specified by three pieces of data:
\begin{itemize}
    \item Let $G=(V,E)$ be an undirected graph, and $q \geq 2$ an integer specifying the number of possible \emph{spins} (or \emph{colors}) each vertex can be assigned.
    \item Let $A \in \R_{\geq0}^{q \times q}$ be a symmetric matrix, often referred to as the \emph{interaction matrix}.\footnote{This can be generalized further by allowing each edge $e \in E$ to have its own interaction matrix $A_{e} \in \R_{\geq0}^{q \times q}$. For simplicity, we will not consider this more general setting in this paper.}
    \item For each $u \in V$, let $\lambda_{u} \in \R_{\geq0}^{q}$ be a nonnegative vector, often referred to as the \emph{external field} applied to $u$.
\end{itemize}
The \emph{Gibbs measure} $\mu$ of the resulting $q$-spin system is a probability distribution over $[q]^{V}$ which assigns each configuration $\sigma \in [q]^{V}$ probability
\begin{align*}
    \mu(\sigma) = \frac{1}{Z} \prod_{uv \in E} A(\sigma_{u},\sigma_{v}) \cdot \prod_{u \in V} \lambda_{u}(\sigma_{u}),
\end{align*}
where
\begin{align*}
    Z \coloneqq \sum_{\sigma \in [q]^{V}} \prod_{uv \in E} A(\sigma_{u},\sigma_{v}) \cdot \prod_{u \in V} \lambda_{u}(\sigma_{u})
\end{align*}
is the \emph{partition function}. We highlight a few emblematic examples.
\begin{itemize}
    \item \textbf{Two-Spin Systems} Fix $q = 2$. Without loss of generality, we may write $A = \begin{bmatrix} \gamma & 1 \\ 1 & \beta \end{bmatrix}$ for $\gamma \geq \beta \geq 0$. For a \emph{fugacity} $\lambda \geq 0$, let $\lambda_{u} = \begin{bmatrix} 1 \\ \lambda \end{bmatrix}$ for all $u \in V$. Then the resulting Gibbs measure $\mu$ is supported on $\{0,1\}^{V}$, and
    \begin{align*}
        \mu_{G}(\sigma) \propto \mathrm{wt}_{G}(\sigma) \coloneqq \beta^{m_{1}(\sigma)} \gamma^{m_{0}(\sigma)} \lambda^{n_{1}(\sigma)},
    \end{align*}
    where $m_{c}(\sigma) = \#\{uv \in E : \sigma_{u} = \sigma_{v} = c\}$ for $c \in \{0,1\}$ counts the number of edges with both endpoints colored $c$, and $n_{1}(\sigma) \coloneqq \#\{u \in V : \sigma_{u} = 1\}$ denotes the number of vertices assigned spin $1$.

    The special case $\beta = 0, \gamma = 1$ recovers the \emph{hardcore model}, since in order for $\mathrm{wt}_{G}(\sigma) > 0$, the set of vertices colored $1$ under $\sigma$ must form an independent set. The special case $\beta = \gamma$ recovers the famous \emph{Ising model}.

    Throughout this paper, by \emph{anti-ferromagnetic}, we mean the setting where $\beta\gamma \leq 1$. Note that this includes the hardcore model.
    \item \textbf{Graph Homomorphisms:} Let $A \in \{0,1\}^{q \times q}$ be the adjacency matrix of a $q$-vertex graph $H$, and set $\lambda_{u} = \allone_{q}$ for all $u \in V$. In this case, adjacent vertices in $G$ are constrained to take spins/colors coming from an edge of $H$, and $\mu$ becomes the uniform distribution over all \emph{graph homomorphisms} from $G$ to $H$ (or \emph{$H$-colorings}). The partition function $Z$ is then the total number of all such $H$-colorings. The special case where $H = K_{q}$ recovers \emph{proper $q$-colorings}.
\end{itemize}

\paragraph{Polymer Models}
Let $G=(V,E)$ be an undirected graph with maximum degree $\Delta$ and $[q]$ be the set of spins. For each vertex $v \in V$, a ground-state spin $g_v$ is assigned. A polymer $\gamma$ is a connected subset $S_\gamma$ of vertices together with an assignment $\sigma^\gamma$ on the connected subset such that $\sigma^\gamma_u \in [q] \setminus g_u$. The size of $\gamma$ is the size of $S_\gamma$, denoted by $\abs{\gamma}$. Two polymers $\gamma_1$ and $\gamma_2$ are compatible if $\mathrm{dist}(S_{\gamma_1},S_{\gamma_2}) \ge 2$. 

Given a collection of allowed polymers $\+C(G)$ together with weight function $w: \+C(G) \to \mathbb{R}_{\ge 0}$, the distribution is supported over the collection of pairwise compatible polymers $\Gamma \subseteq \+C(G)$. Specifically, given $\Gamma \subseteq \+C(G)$ such that for all $\gamma_1 \neq \gamma_2 \in \Gamma$ satisfying $\gamma_1$ is compatible with $\gamma_2$, the weight of drawing $\Gamma$ is proportional to $\prod_{\gamma \in \Gamma} w_\gamma$, i.e.,
\begin{align*}
    \mu(\Gamma) \propto \prod_{\gamma \in \Gamma} w_\gamma. 
\end{align*}
The partition function is the total of weight $\prod_{\gamma \in \Gamma} w_\gamma$ over all pairwise compatible polymers $\Gamma$.

\paragraph{Hypergraph Independent Sets}
Let $G=(V,E \subseteq 2^{V})$ be a hypergraph. An \emph{independent set} of the hypergraph is a subset of vertices $I \subseteq V$ such that $e \not\subset I$ for all hyperedges $e \in E$. The corresponding Gibbs distribution is uniform over all hypergraph independent sets, and the corresponding partition function counts the number of such independent sets.

\subsection{Uniqueness}
We will use the following notion of uniqueness for anti-ferromagnetic two-spin systems \cite{li2013correlation}.
\begin{definition}\label{def:d-unique}
    Let $d \geq 1$ be an integer and $\delta \in (0,1)$ be a real number.
    Let $\beta,\gamma,\lambda$ be parameters such that $0 \leq \beta \leq \gamma$, $\beta \gamma < 1$, and $\gamma,\lambda$ are positive.
    We say $(\beta, \gamma, \lambda)$ is \emph{$d$-unique with gap $\delta$} if 
    \begin{align*}
        \abs{F_d'(\hat{x}_d)} = \frac{d(1-\beta\gamma) \hat{x}_d}{(\beta \hat{x}_d + 1)(\hat{x}_d + \gamma)} \leq 1-\delta,
    \end{align*}
    where
    \begin{align*}
        \hat{x}_d \text{ is the unique fixed point of } F_d(x) \coloneqq \lambda\tp{\frac{\beta x + 1}{x + \gamma}}^d.
    \end{align*}
    Moreover if $(\beta,\gamma,\lambda)$ is $d$-unique with gap $\delta$ for all $1 \leq d < \Delta$, then we say $(\beta,\gamma,\lambda)$ is \emph{up-to-$\Delta$ unique with gap $\delta$}.
\end{definition}
It is now well-known that for anti-ferromagnetic two-spin systems, up-to-$\Delta$ uniqueness implies the existence of $\mathsf{FPRAS}$ for the partition function for arbitrary graphs of maximum degree $\Delta$ \cite{weitz2006counting, li2013correlation}. Conversely, if the parameters do not satisfy up-to-$\Delta$ uniqueness, then no $\mathsf{FPRAS}$ exists unless $\mathsf{NP} = \mathsf{RP}$ \cite{Sly10, SS12, GSV16}.

\subsection{Computational tasks} \label{sec:tasks}
In this section, we give precise definitions for the computational tasks considered in this paper.
Let $G = (V,E)$ be a graph and $\mu_G$ be the Gibbs distribution of any model under consideration on $[q]^V$ for some integer $q > 0$.

\begin{definition}[perfect marginal sampling] \label{def:general-marginal-sampling}
    \ 
    \begin{itemize}[leftmargin=10ex]
        \item[\bf \textsf{Input}] essential parameters (e.g. fugacity $\lambda$ for the hardcore model) and a vertex $u \in V$; 
        \item[\bf \textsf{Oracle}] 
        oracle $\textsf{neighbors}(v)$ will return a list $\Gamma_v$ of neighbors of $v$ in graph $G$; \\
        oracle $\textsf{pinning}(v)$ will return $\tau_v$ when $v \in \Lambda$ and will return $\perp$ otherwise; 
        \item[\bf \textsf{Output}] a random spin configuration $\sigma \sim \mu^\tau_{G,u}$.
    \end{itemize}
\end{definition}
\begin{remark}
Throughout this paper, most perfect marginal samplers we work with will run in (expected) constant time, though in some cases subpolynomial time suffices (see \Cref{rmk:marginal-sampler-time}).
\end{remark}

\begin{definition}[fully polynomial-time randomized approximation scheme ($\mathsf{FPRAS}$)] \label{def:FPRAS}
    \ 
    \begin{itemize}[leftmargin=10ex]
        \item[\bf \textsf{Input}] essential parameters, a graph $G = (V,E)$, and an error tolerance $\epsilon > 0$;
        \item[\bf \textsf{Output}] a random number $\widehat{Z}$ such that 
        \[\Pr{\e^{-\epsilon} Z_G \leq \widehat{Z} \leq \e^{\epsilon} Z_G} \geq \frac{3}{4}.\]
    \end{itemize}
\end{definition}

\subsection{Reduction from counting to sampling}

The following result is the standard reduction from randomized approximate counting to randomized inference.
We also provide a proof for completeness.
\begin{lemma}[\text{\cite{jerrum1986random}}] \label{lem:reduction-counting-to-inference}
    Let $\mu$ be a distribution on $[q]^n$. 
    Let $\epsilon \in (0,1)$ be a real number.
    Suppose there is an algorithm $\mathcal{A}$ such that for any $v \in [n], c \in [q]$, any $\Lambda \subseteq [n]$, and any pinning $\tau \in [q]^\Lambda$ such that $\mu_\Lambda(\tau) > 0$, $\mathcal{A}$ can compute a random number $P \geq 0$ such that $\E{P} = \mu^\tau_v(c)$ and $\Var{P}/\E{P}^2 \leq \epsilon^2/n$, in time $T(\epsilon, n)$.
    Then, there is an algorithm that for any feasible $\sigma \in [q]^n$ with $\mu(\sigma) > 0$, it computes a random number $M$ in time $O(n) \cdot T(\epsilon, n)$ such that
    \begin{align*}
        \Pr{\e^{-\epsilon} \mu(\sigma) \leq M \leq \e^\epsilon \mu(\sigma)} \geq \frac{3}{4}.
    \end{align*}
\end{lemma}

\subsection{Sampling oracle for binomial distribution}

\begin{lemma}[\cite{colton15exact}]\label{lem:binomial}
    There exists a data structure that, after an $\widetilde{O}(\sqrt{N})$ initialization phase, enables exact sampling from the binomial distribution $\mathrm{Bin}(n,p)$ for any $n \le N$ and $p \in [0,1]$ in $O(\log^2 N)$ time in expectation per query.
\end{lemma}

As a corollary, we are able to sample from the multinomial distribution $M(n;\nu)$ within poly-logarithmic time when $\nu$ is supported on a constant-size set.

\subsection{Enumerating connected subgraphs}
\begin{lemma}[\cite{BCKL13,PR17}]\label{lem:connected-subgraph}
    Let $G=(V,E)$ be an undirected graph with maximum degree $\Delta \ge 3$. 
    \begin{itemize}
    \item The number of size-$k$ connected subgraph containing a vertex $u \in V$ is bounded by $(\e \Delta)^{k-1}$;
    \item Moreover, there exists an algorithm enumerating all size-$k$ connected subgraphs containing a vertex $u \in V$ in $O\tp{k^7 (\e \Delta)^{2k}}$ time.
    \end{itemize}
\end{lemma}

\section{Black-box reduction for two-spin systems} \label{sec:2-spin-reduction}
In this section, we establish the black-box reduction from subquadratic-time counting to constant-time perfect marginal sampling for anti-ferromagnetic two-spin systems stated in \Cref{thm:2-spin-reduction}.
Note that according to \Cref{lem:reduction-counting-to-inference}, \Cref{thm:2-spin-reduction} is a direct corollary of the following result.

\begin{theorem} \label{thm:2-spin-reduction-inference}
    Let $\delta > 0$, $\epsilon \in (0,1)$ be real numbers, and $\Delta \geq 1$ be a constant integer.
    Fix $\beta \geq 0, \gamma > 0, \lambda > 0$, $\beta\gamma < 1$ such that $(\beta, \gamma, \lambda)$ are up-to-$\Delta$ unique with gap $\delta$.
    If there is an algorithm $\mathcal{A}$ such that for any graph with max degree $\Delta$, $\mathcal{A}$ can finish the perfect marginal sampling task in \Cref{def:general-marginal-sampling} in expected time $O(1)$, then there is an algorithm that, for any $n$-vertex graph $G = (V,E)$ with max degree $\Delta$, any $\Lambda \subseteq V$, any pinning $\tau \in \Omega(\mu_{G,\Lambda})$, and any vertex $v \in V$, it computes a random number $M$ in expected time $\widetilde{O}((n/\epsilon^2)^{1 - (\log \frac{1}{1-\delta})/ \log \frac{\Delta}{1-\delta}})$ such that 
    \begin{align*}
        \E{M} = \mu^\tau_v(1) \quad \text{and} \quad \frac{\Var{M}}{\E{M}^2} \leq \epsilon^2/n.
    \end{align*}
\end{theorem}

Our reduction relies heavily on the existence of the self-avoiding walk tree (SAW tree) introduced by Weitz~\cite{weitz2006counting}.
Here, we use an algorithmic version of it.
\begin{lemma}[\text{\cite{weitz2006counting}}] \label{lem:SAW-tree}
Fix an $n$-vertex graph $G = (V, E)$ with max degree $\Delta$ and parameters $\beta \geq 0$, $\gamma, \lambda > 0$.
Let $\mu_G$ be the Gibbs distribution of the two-spin system.
Then, for every vertex $v \in V$, there is a SAW tree $\mathbb{T}_{\mathrm{SAW}} = \mathbb{T}_{\mathrm{SAW}}(G,v) = (V_{\mathrm{SAW}}, E_{\mathrm{SAW}})$ with the same max degree and $v \in V_{\mathrm{SAW}}$, such that:
\begin{itemize}
    \item for all $\Lambda \subseteq V$ and $\tau \in \Omega(\mu_\Lambda)$, there is $\tau_{\mathrm{SAW}}$ on some $\Lambda_{\mathrm{SAW}} \subseteq V_{\mathrm{SAW}}$ that $\mu_{\mathbb{T}_{\mathrm{SAW}}, v}^{\tau_{\mathrm{SAW}}} = \mu_{G,v}^\tau$;
    \item there is a data structure $\mathcal{D}$ which maintains a vertex $u \in V_{\mathrm{SAW}}$ and supports the  following operations in $O_\Delta(\mathrm{dist}_{\mathbb{T}_{\mathrm{SAW}}}(v, u))$ time:
    \begin{enumerate}
        \item $\mathcal{D}.\mathsf{neighbors}()$ returns a list of neighbors of $u$; 
        \item $\mathcal{D}.\mathsf{move}(w)$ updates $u$ to $w$ if $w$ is a neighbor of $u$.
    \end{enumerate}
\end{itemize}
\end{lemma}


\newcommand{\dep}{\mathsf{depth}}
\newcommand{\siz}{\mathsf{size}}

\begin{definition} \label{def:boundary-on-tree}
    Let $\mathbb{T} = (V,E)$ be a tree rooted at $v \in V$.
    We say $S \subseteq V$ is a \emph{boundary} of $\mathbb{T}$ if for all $u \in S$, there is no descendant $u$ in $S$.
    Moreover, we define 
    \begin{align*}
        \siz_{\mathbb{T}}(S) &\coloneqq \abs{S \cup \set{\text{the connected component in $\mathbb{T} - S$ that contains $v$}}}, \\
        \dep_{\mathbb{T}}(S) &\coloneqq \max_{u \in S} \mathrm{dist}_{\mathbb{T}}(v, u).
    \end{align*}
\end{definition}

\begin{lemma}[SAW tree with flower] \label{lem:SAW-tree-with-flower}
    Let $G = (V, E)$ be a graph with max degree $\Delta$ and parameters $\beta \geq 0$, $\gamma, \lambda > 0$.
    Let $\mu_G$ be the Gibbs distribution of the two-spin system.
    Fix $v \in V$. Let $\mathbb{T}_{\mathrm{SAW}} = \mathbb{T}_{\mathrm{SAW}}(G, v) = (V_{\mathrm{SAW}}, E_{\mathrm{SAW}})$ be the SAW tree rooted at $v$.
    Then for any $S \subseteq V_{\mathrm{SAW}}$ be a boundary of $\mathbb{T}_{\mathrm{SAW}}$, there is a graph $G_{\mathrm{SF}} = G_{\mathrm{SF}}(G, v, S) =  (V_{\mathrm{SF}}, E_{\mathrm{SF}})$ obtained from $\mathbb{T}_{\mathrm{SAW}}$ by, for every $u \in S$, replacing the sub-tree $\mathbb{T}_u$ rooted at $u$ to some graph $G_u$ that contains $u$ and keeps its degree, such that $G_{\mathrm{SF}}$ enjoys the following properties:
    \begin{itemize}
        \item for $\Lambda \subseteq V$ and $\tau \in \Omega(\mu_\Lambda)$, there is $\Lambda_{\mathrm{SF}} \subseteq V_{\mathrm{SF}}$ and a pinning $\tau_{\mathrm{SF}}$ on $\Lambda_{\mathrm{SF}}$ that $\mu^\tau_{G, v} = \mu^{\tau_{\mathrm{SF}}}_{G_{\mathrm{SF}}, v}$.
        \item there is a data structure $\mathcal{O}$ such that after $O_\Delta(\dep_{\mathbb{T}_{\mathrm{SAW}}}(S) \cdot \siz_{\mathbb{T}_{\mathrm{SAW}}}(S))$ time initialization, it can answer the following queries in time $O_\Delta(\dep_{\mathbb{T}_{\mathrm{SAW}}}(S))$: for every $u \in V_{\mathrm{SF}}$,
        \begin{enumerate}
            \item $\mathcal{O}.\textsf{neighbors}(u)$ will return a list $\Gamma_u$ of neighbors of $u$ in graph $G_{\mathrm{SF}}$; 
            \item $\mathcal{O}.\textsf{pinning}(u)$ will return $\tau_{\mathrm{SF},u}$ when $u \in \Lambda_{\mathrm{SF}}$ and will return $\perp$ otherwise; 
        \end{enumerate}
    \end{itemize}
    We call such graph $G_{\mathrm{SF}}$ a \emph{\underline{S}AW tree with \underline{f}lower}.
\end{lemma}

The proof of \Cref{lem:SAW-tree-with-flower} is deferred to \Cref{sec:SAW-tree-with-flower}.

We now present our algorithm, which is similar to the one used in \cite[Theorem 1.1]{AFFGW25} at a high level.
Fix $\Lambda \subseteq V$, $\tau \in \Omega(\mu_{G,\Lambda})$, and $v \in V$. 
Our algorithm proceeds as follows:
\begin{enumerate}[label={\sf Step-\arabic*}:, ref={\sf Step-\arabic*}, leftmargin=2cm]
    \item\label{step:1} We first choose a boundary $S \subseteq V_{\mathrm{SAW}}$ on the SAW tree $\mathbb{T}_{\mathrm{SAW}}(G,v)$ according to \Cref{algo:boundary-on-tree} with $\delta$ being the uniqueness gap and $N = \widetilde{O}(n)$. We prepare the data structure $\mathcal{O}$ for the SAW tree with flower $G_{\mathrm{SF}}(G,v,S)$ by \Cref{lem:SAW-tree-with-flower}.
    \item\label{step:2} Then, we generate $\sigma \sim \mu^{\tau_{\mathrm{SF}}}_{G_{\mathrm{SF}}, S}$ by calling the black-box algorithm for perfect marginal sampling (see \Cref{def:general-marginal-sampling}) $\abs{S}$ times with oracles provided by $\mathcal{O}$.
In particular, suppose $S = \set{u_1, \cdots, u_{\abs{S}}}$ and let $\sigma( < i) \coloneqq \sigma(u_1, \cdots, u_{i-1})$, we use the black-box perfect marginal sampling algorithm to sample  $\sigma_i \sim \mu^{\tau_{\mathrm{SF}}, \sigma(< i)}_{G_{\mathrm{SF}}, u_i}$  for $i = 1, \cdots, \abs{S}$.
    \item\label{step:3} Finally, we calculate and output $M \coloneqq \mu^{\tau_{\mathrm{SF}}, \sigma}_{G_{\mathrm{SF}}, v}(1)$ by the tree recursion.
Note that this gives an unbiased estimator because $\E{M} = \oE_{\sigma}[\mu^{\tau_{\mathrm{SF}}, \sigma}_{G_{\mathrm{SF}}, v}(1)] = \mu^\tau_v(1)$.
\end{enumerate}

\begin{algorithm}
\caption{\label{algo:boundary-on-tree} Boundary on tree: $\mathrm{boundary}_{\mathbb{T}, \delta}(v, N)$}
 \KwIn{parameter $\delta \in (0,1)$, $N  > 0$, $\mathbb{T} = (V,E)$ with root $v$}
 \KwOut{a boundary $S \subseteq V$ of $\mathbb{T}$}
 $\mathrm{boundary}_{\mathbb{T}, \delta}(u, N) \coloneqq $ \Begin{
    let $C$ be the set of children of $u$ in $\mathbb{T}$ with $d = \abs{C}$\;
    \lIf{$N \leq 1$}{\Return $\set{u}$ }
    \lIf{$C = \emptyset$}{\Return $\emptyset$}
    \lFor{$w \in C$}{
        let $S_w = \mathrm{boundary}_{\mathbb{T}_w, \delta}(w, N \cdot \frac{1-\delta}{d})$
    }
    \Return $\bigcup_{w \in C} S_w$ \;
 }
\end{algorithm}

Compared with \cite{AFFGW25}, the main differences are:
(1) we use the SAW tree with flower (\Cref{lem:SAW-tree-with-flower}) instead of constructing SAW tree on-the-fly to remove further requirements on perfect marginal samplers and achieve a black-box reduction;
(2) inspired by the \emph{computationally efficient correlation decay}~\cite{li2012approximate, li2013correlation}, we use a more sophisticated strategy to truncate the SAW tree (see \Cref{algo:boundary-on-tree}), which makes the reduction work for general anti-ferromagnetic two-spin systems (including non-monotone $\gamma > 1$ cases);
(3) we provide a tighter variance analysis that works all the way up to uniqueness.

\begin{remark}
    We note that in the rest of this section, we will focus on the case $N = \tilde{O}(n/\epsilon^2)$.
    Intuitively, $N^{-1}$ has the same magnitude as the desired variance bound.
\end{remark}

The crucial reason for sub-linear running time is that the sub-tree explored by our algorithm is of sub-linear size.
Specifically, if $G$ is a $\Delta$-regular graph, the boundary $S$ consists of vertices at level $O(\log_{\frac{\Delta-1}{1-\delta}}(N))$ of the SAW tree.
\begin{proposition} \label{prop:size-of-boundary}
    Let $\Delta = O(1)$, $\delta \in (0,1)$, and $N \geq 1$ be parameters.
    Fix a tree $\mathbb{T}$ with root $v$ and maximum degree $\Delta$.
    Let $S = \mathrm{boundary}_{\mathbb{T}, \delta}(v, N)$ as in \Cref{algo:boundary-on-tree}.
    It holds that
    \begin{align*}
        \abs{S} \leq T(N) = \widetilde{O}\left(N^{1 - (\log \frac{1}{1-\delta})/ \log \frac{\Delta}{1-\delta}}\right) \quad \text{and} \quad \dep_{\mathbb{T}}(S) \leq O\left(\log_{\frac{\Delta-1}{1-\delta}}(N)\right), 
    \end{align*}
    where $T(N)$ is the maximum number of vertices in $\mathrm{boundary}_{\mathbb{T}, \delta}(v, N)$ over all $\mathbb{T}$ and $v$.
\end{proposition}
\begin{remark}
One can sharpen the bound on $T(N)$ slightly to $\widetilde{O}\left(N^{1 - (\log \frac{1}{1-\delta})/ \log \frac{\Delta-1}{1-\delta}}\right)$ because all internal vertices of the SAW tree (besides the root vertex) have at most $\Delta - 1$ children, not $\Delta$. However, we establish a slightly weaker bound to keep the proof as clean as possible.
\end{remark}
\begin{proof}[Proof]
    The bound on $\dep_{\mathbb{T}}(S)$ is direct.
    We will focus on bounding $\abs{S} \leq T(N)$.
    According to \Cref{algo:boundary-on-tree}, we choose an upper bound $T$ that satisfies the following recursion:
    \begin{align*}
        T(x) &:= \begin{cases}
           \max_{1 \leq d \leq \Delta} d \cdot T\tp{x \cdot \frac{1-\delta}{d}}, &\text{if } x > 1, \\
           1, &\text{if } x \leq 1.
        \end{cases}
    \end{align*}
    Equivalently, we can also define this upper bound as follows
    \begin{align*}
        T(x) 
        &:= \max\set{\prod_{i=1}^k d_i \,\, \Bigg| \,\, k \geq 1 \land d_1, \cdots, d_k \in [1,\Delta] \land x \frac{(1-\delta)^{k-1}}{\prod_{i=1}^{k-1} d_i} > 1 \geq x \frac{(1-\delta)^{k}}{\prod_{i=1}^{k} d_i} }.     
    \end{align*}
    Intuitively, we enumerate every feasible way of reducing $x$ to be some point $\leq 1$.
    And, we choose $T(x)$ to be the worst case over all these possibilities.
    We define $k^\star = k^\star(x)$ as follows
    \begin{align} \label{eq:def-k-star}
        k^\star(x) := \min\set{k \geq 0 \,\, \Bigg| \,\, \exists d_1, \cdots, d_k \in[1,\Delta], \text{ s.t. } \prod_{i=1}^k d_i \geq (1-\delta)^k x}.
    \end{align}
    Then, we relax $T(x)$ as
    \begin{align}
    \nonumber
        T(x) 
        &\leq \Delta \cdot \max\set{\prod_{i=1}^{k-1} d_i \,\, \Bigg| \,\, k \geq 1 \land d_1, \cdots, d_k \in [1,\Delta] \land x \frac{(1-\delta)^{k-1}}{\prod_{i=1}^{k-1} d_i} > 1 \geq x \frac{(1-\delta)^{k}}{\prod_{i=1}^{k} d_i}} \\
        \label{eq:upper-bound-Tx}
        &\leq \Delta \cdot x (1-\delta)^{k^\star - 1},
    \end{align}
    where the first inequality holds by $d_k \leq \Delta$ and the last inequality follows from the constraint $\prod_{i=1}^{k-1} d_i \leq x(1-\delta)^{k-1}$.
    By taking $d_1 = \cdots = d_k = \Delta$ in \eqref{eq:def-k-star}, we have the following bound on $k^\star$:
    \begin{align}\label{eq:upper-bound-k-star}
        x(1-\delta)^{k^\star} \leq \Delta^{k^\star} \quad \Longleftrightarrow \quad k^\star \leq \ftp{\log(x) / \log \frac{\Delta}{1-\delta}} \leq \log(x) / \log \frac{\Delta}{1-\delta} + 1.
    \end{align}
    Finally, let $x = N$.
    Combining \eqref{eq:upper-bound-k-star} and \eqref{eq:upper-bound-Tx} finishes the proof 
    \begin{align*}
        T(x) &\leq \Delta \cdot x^{1 - \log(\frac{1}{1-\delta}) / \log \frac{\Delta}{1-\delta}}. \qedhere
    \end{align*} 
\end{proof} 

Now, we are able to justify the running time of our algorithm. 
\begin{proof}[Proof of runtime of \Cref{thm:2-spin-reduction-inference}]
    Recall our algorithm is described as in \ref{step:1} to \ref{step:3}.
    We set $N = \tilde{O}(n/\epsilon^2)$.
    Note that \Cref{algo:boundary-on-tree} can be implemented by using the data structure $\mathcal{D}$ provided by \Cref{lem:SAW-tree}.
    Then by \Cref{prop:size-of-boundary}, the overall runtime of \Cref{algo:boundary-on-tree} will be $\widetilde{O}(T(N))$, where $\widetilde{O}$ hides $\mathrm{polylog}(n)$ factors.
    Again by \Cref{prop:size-of-boundary} and \Cref{lem:SAW-tree-with-flower}, the data structure $\mathcal{O}$ can also be initialized within time $\widetilde{O}(T(N))$.
    Similarly, by \Cref{prop:size-of-boundary} and \Cref{lem:SAW-tree-with-flower}, the \ref{step:2} and \ref{step:3} of will only costs $\widetilde{O}(T(N))$ time.
    Hence, overall, the algorithm runs in time $\widetilde{O}(T(N)) = \widetilde{O}((n/\epsilon^2)^{1 - (\log \frac{1}{1-\delta})/ \log \frac{\Delta}{1-\delta}})$.
\end{proof}

In the rest of this section, we show that for an appropriate choice of $N = O(n/\epsilon^2)$, the random real number $P$ returned by the algorithm (\ref{step:1} to \ref{step:3}) satisfies $\Var{M} \leq \epsilon^2/n$.
The following notion of correlation decay for two-spin systems plays a major role in our analysis.

\newcommand{\infsq}{\mathrm{inf}^2}
\begin{definition}[total squared influence]\label{def:total-squared-influence}
  Let $V$ be the ground set and $\mu$ be a distribution on $\set{0,1}^V$.
  We say $\mu$ has \emph{total-squared-influence} $C \geq 0$ from $v \in V$ to $S\subseteq V$, if for any $\Lambda \subseteq V$ and $\sigma \in \Omega(\mu_\Lambda)$ such that $\mu^\sigma_v(1) \in (0,1)$, it holds that:
  \begin{align*}
    \infsq_\mu(v \to S) \coloneqq \sum_{u \in S} \left(\mu^{\sigma, v\gets 1}_u(1) - \mu^{\sigma,v\gets 0}_u(1)\right)^2 \leq C.
  \end{align*}
\end{definition}

\begin{definition}[marginal bound] \label{def:marginal-lower-bound}
    Let $b > 0$ and $V$ be the ground set. 
    We say a distribution $\mu$ on $\set{0,1}^V$ is $b$-marginally bounded if for any $v \in V$, $c \in \set{0,1}$, $\Lambda \subseteq V$, and pinning $\tau \in \Omega(\mu_\Lambda)$, we have $\mu^\tau_v(c) > 0$ implies $\mu^\tau_v(c) \geq b$.
\end{definition}

\begin{lemma} \label{lem:var-bound-squared-influence}
    Let $\mu$ be a distribution on $\set{0,1}^V$ for a finite ground set $V$.
    Let $\Lambda \subseteq V$, and $\tau \in \Omega(\mu_\Lambda)$ be a pinning.
    Let $v \in V$, $S\subseteq V$ such that for all $\sigma \in \Omega(\mu^\tau_{S})$, we have $\mu^{\tau\cup \sigma}_v(1) \in (0,1)$.
    Let $b > 0, C \geq 0$.
    If $\mu$ is $b$-marginally bounded and has total-squared-influence $C$ from $v$ to $S$, then the function $f(X) \coloneqq \mu^{\tau \cup X(S \setminus \Lambda)}_v(1)$ satisfies
    \begin{align*}
        \E[X \sim \mu^{\tau}]{f(X)} &= \mu_{v}^{\tau}(1) \\
        \Var[X \sim \mu^{\tau}]{f(X)} &\leq \log\tp{\abs{S\setminus\Lambda}} \cdot (2b)^{-2} \cdot C.
    \end{align*}    
\end{lemma}
The proof of \Cref{lem:var-bound-squared-influence} is given in \Cref{sec:var-bound-squared-influence}. We remark that by Chebyshev's inequality, \Cref{lem:var-bound-squared-influence} gives a criterion for the concentration of the conditional marginal $f(X) = \mu_{v}^{\tau \cup X(S\setminus \Lambda)}$, which may be of independent interest.

\begin{lemma} \label{lem:total-squared-influence-2-spin}
    Let $\delta > 0$ and $\Delta \geq 1$ be a constant integer.
    Fix $\beta \geq 0, \gamma > 0, \lambda > 0$, $\beta\gamma < 1$ such that $(\beta, \gamma, \lambda)$ are up-to-$\Delta$ unique with gap $\delta$.
    Let $G = (V,E)$ be a graph with max degree $\Delta$.
    Fix $v \in V$. Let $\mathbb{T} = \mathbb{T}_{\mathrm{SAW}}(v, G)$ and $S = \mathrm{boundary}_{\mathbb{T}, \delta}(v, N)$.
    Let $G_{\mathrm{SF}} = G_{\mathrm{SF}}(G, v, S) = (V_{\mathrm{SF}}, E_{\mathrm{SF}})$ be the SAW tree with flower , and $\nu = \mu_{G_{\mathrm{SF}}}$ be the Gibbs distribution.
    Let $\Lambda \subseteq V_{\mathrm{SF}}$ and $\tau \in \Omega(\nu_\Lambda)$ such that $\nu^\tau_v(1) \in (0,1)$.
    Then, for $N \geq 1$,
    \begin{align*}
        \infsq_{\nu^\tau}(v \to S) \leq C_{\Delta,\delta} \cdot N^{-1},
    \end{align*}
    where $C_{\Delta,\delta}$ is a constant that only depends on $\Delta$ and $\delta$.
\end{lemma}

The proof of \Cref{lem:total-squared-influence-2-spin} is provided in \Cref{sec:total-squared-influence-2-spin-on-tree}.
Now, we are ready to bound $\Var{M}$ and finish the proof of \Cref{thm:2-spin-reduction-inference}.
\begin{proof}[Proof of variance bound in \Cref{thm:2-spin-reduction-inference}]
    Let $M$ with $\E{M} = \mu^\tau_v(1)$ be the random number returned by the algorithm (\ref{step:1} to \ref{step:3}).
    Without loss of generality, we assume that $\mu^\tau_v \in (0,1)$, since, otherwise, we have $\Var{M} = 0$.
    Let $\mathbb{T} = \mathbb{T}_{\mathrm{SAW}}(v, G)$ and $S = \mathrm{boundary}_{\mathbb{T}, \delta}(v, N)$.
    Let $G_{\mathrm{SF}} = G_{\mathrm{SF}}(G, v, S) = (V_{\mathrm{SF}}, E_{\mathrm{SF}})$ be the SAW tree with flower , and $\nu = \mu_{G_{\mathrm{SF}}}$ be the Gibbs distribution.
    Let $\Lambda_{\mathrm{SF}} \subseteq V_{\mathrm{SF}}$ and $\tau_{\mathrm{SF}}$ be a pinning on $\Lambda_{\mathrm{SF}}$ that $\mu^\tau_{G, v} = \mu^{\tau_{\mathrm{SF}}}_{G_{\mathrm{SF}},v} = \nu^{\tau_{\mathrm{SF}}}_v$.
    We assume that $N$ is larger than some constant to make sure that $\mathrm{dist}(v, S) \geq 2$ so that for all $\sigma \in \Omega(\nu^{\tau_{\mathrm{SF}}}_{S})$, we have $\nu^{\tau_{\mathrm{SF}} \cup \sigma}_v(1) \in (0,1)$.
    Then, by \Cref{lem:var-bound-squared-influence} and \Cref{lem:total-squared-influence-2-spin}, we have
    \begin{align*}
        \Var{M} / \E{M}^2 
        &\leq  b^{-4} \cdot C_{\Delta,\delta} \cdot N^{-1} \cdot \log \abs{S} \\
        &\leq b^{-4} \cdot C_{\Delta,\delta} \cdot N^{-1} \cdot \log n,
    \end{align*}
    where by \Cref{prop:size-of-boundary}, the last inequality holds when $n$ is larger than some universal constant (when $n$ is constant, $\mu^\tau_v(1)$ can be computed deterministically via brute force).
    We achieve the desired variance bound by taking $N = b^{-4} \cdot C_{\Delta,\delta} \cdot n \log n / \epsilon^2$.
    We finish the proof by noticing that $b = O(1)$ for constant $\Delta, \lambda, \beta, \gamma$.
\end{proof}

\subsection{Concentration via bounded total-squared-influence} \label{sec:var-bound-squared-influence}
In this section, we prove \Cref{lem:var-bound-squared-influence}.
Our proof approach is inspired by the local-to-global analysis in~\cite{chen2025faster}.
We consider the following denoising process, which is a modification of the coordinate-by-coordinate denoising process.

\begin{definition} \label{def:denoising}
  Recall sets $S, \Lambda \subseteq V$, distribution $\mu$ on $\set{0,1}^V$, pinning $\tau \in \set{0,1}^\Lambda$, and $X\sim \mu^\tau$ from the statement of \Cref{lem:var-bound-squared-influence}.
  Let $m = \abs{S \setminus \Lambda}$.
  We use the Markov chain $(X_t, \Lambda_t)_{t=0}^m$ to generate $X \sim \mu^\tau$, which is defined as follows:
  \begin{enumerate}
  \item $X_0 = \tau$, $\Lambda_0 = \Lambda$;
  \item given $X_t \in \set{0,1}^{\Lambda_t}$, get $X_{t+1}, \Lambda_{t+1}$ by
    \begin{enumerate}
    \item draw $u \in S \setminus \Lambda_t$ u.a.r;
    \item draw $c \sim \mu^{X_t}_{u}$;
    \item let $X_{t+1} \gets X_t \cup \set{u \gets c}$;
    \item let $\Lambda_{t+1} \gets \Lambda_t \cup \set{u}$.
    \end{enumerate}
  \end{enumerate}
  Finally, sample $X \sim \mu^{X_m}$ which implies $X \sim \mu^\tau$.
\end{definition}

\begin{proof}[Proof of \Cref{lem:var-bound-squared-influence}]
  For each $t$, define $\oVar(t) \coloneqq \E{\Var{f(X) \mid X_t}}$, where $X$ is drawn according to \Cref{def:denoising}. Also, recall that we define $f(X) \coloneqq \mu^{\tau \cup X(S \setminus \Lambda)}_v(1)$.
  According to this definition, we have $\oVar(0) = \Var{f(X)}$; and $\oVar(m) = 0$ because $f$ only depends on the $m$ coordinates of $X$ in $S \setminus \Lambda_0$.
  Moreover, letting $F \coloneqq f(X)$, we have by telescoping that
  \begin{align*}
    \Var{F}
    &= \sum_{t=0}^{m-1} \tp{\oVar(t) - \oVar(t+1)} \\
    &= \sum_{t=0}^{m-1} \tp{ \E{\Var{F \mid X_t}} - \E{\Var{F \mid X_{t+1}}} } \\
    &= \sum_{t=0}^{m-1} \E{\Var{F \mid X_t} - \E{\Var{F \mid X_{t+1}} \mid X_t}} \\
    &= \sum_{t=0}^{m-1} \E{\Var{\E{F \mid X_{t+1}} \mid X_t}},
  \end{align*}
  where the last equation holds by the law of total variance.
  For convenience, let $S_t \coloneqq S \setminus \Lambda_t$.
  Fix $0\leq t \leq m-1$, and $X_t$, then it holds that
  \begin{align*}
    \Var{\E{F \mid X_{t+1}} \mid X_t} 
    &= \frac{1}{\abs{S_t}} \sum_{u \in S_t} \Var{\E{F \mid X_t \land X_u} \mid X_t} \\
    &\leq \frac{1}{\abs{S_t}} \sum_{u \in S_t} \tp{\E{\mu^{\tau \cup X(S_0) }_v(1) \,\, \Big| \,\, X_t, X_u \gets 1} - \E{\mu^{\tau \cup X(S_0)}_v(1) \,\, \Big| \,\, X_t, X_u \gets 0}}^2 \\
    \tag{law of total expectation} &= \frac{1}{\abs{S_t}} \sum_{u \in S_t} \tp{\mu^{X_t, u\gets 1}_v(1) - \mu^{X_t, u\gets 0}_v(1)}^2 \\
    &\leq \frac{1}{4b^{2}\abs{S_t}} \sum_{u \in S_t} \tp{\mu^{X_t, v\gets 1}_u(1) - \mu^{X_t, v\gets 0}_u(1)}^2 ,
  \end{align*}
  where the last inequality holds because of the following identity using Bayes' Rule
  \begin{align*}
      \mu_{u}^{v \gets 1}(1) - \mu_{u}^{v \gets 0}(1) = \frac{\mu_{u}(1)\mu_{u}(0)}{\mu_{v}(1)\mu_{v}(0)} \cdot (\mu_{v}^{u \gets 1}(1) - \mu_{v}^{u \gets 0}(1)),
  \end{align*}
  where the numerator $\mu_{u}(1)\mu_{u}(0)$ is at most $\sup_{p} p(1-p) \leq 1/4$ while the denominator $\mu_{v}(1)\mu_{v}(0)$ is at least $b(1-b) \geq b/2$.
  Then, by the total-squared-influence, we have
  \begin{align*}
    \Var{\E{F \mid X_{t+1}} \mid X_t} &\leq \frac{C}{4b^{2}\abs{S_t}} . 
  \end{align*}
  Summing over $t$, we obtain
  \begin{align*}
    \Var{F} &\leq \frac{C}{4b^{2}} \sum_{t=0}^{m-1} \frac{1}{m-t} \leq \frac{C \log m}{4b^{2}}. \qedhere
  \end{align*}
\end{proof}

\subsection{Total-squared-influence for two-spin systems on tree} \label{sec:total-squared-influence-2-spin-on-tree}

In this section, we prove \Cref{lem:total-squared-influence-2-spin}.
We use the following standard facts concerning two-spin systems.

\begin{lemma}[\text{\cite[Lemma B.2]{anari2020spectral}}]\label{lem:product}
  Let $G = (V, E)$ be a graph and $\mu = \mu_G$ be the Gibbs distribution of a two-spin system on $G$.
  Let $a,b,c$ be three vertices in $G$. If the removal of $b$ disconnects $a$ and $c$, then
  \[\infsq_\mu(a \to c) = \infsq_\mu(a\to b) \cdot \infsq_\mu(b\to c).\]
\end{lemma}

For every integer $d \geq 1$, recall the multivariate tree recursion $F_d:\mathbb{R}^d \to \mathbb{R}$ is defined as
\begin{align*}
  F_d(\*x) \coloneqq \lambda \prod_{i=1}^d\frac{\beta x_i + 1}{x_i + \gamma}.
\end{align*}

\begin{lemma}[\text{\cite[Lemma 16]{chen2020rapid}}] \label{lem:formula-influence}
    Let $G = (V, E)$ be a graph and $\mu = \mu_G$ be the Gibbs distribution of a two-spin system on $G$.
    Let $v \in V$ be a vertex of degree $d$, $\Lambda \subseteq V \setminus\set{v}$ and $\tau \in \Omega(\mu_\Lambda)$.
    If the removal of $v$ disconnects every vertices in $N(v)$, we have
    \begin{itemize}
        \item $R^\tau_{G, v} = F_d((\*R^\tau_{G-v,w})_{w\in N(v)})$, where $R^\tau_{G, v} \coloneqq \frac{\mu^\tau_{G,v}(1)}{\mu^\tau_{G,v}(0)}$ is the marginal ratio at $v$;
        \item for $w \in N(v)$, $\infsq_{\mu^\tau_G}(v \to w) = h(R^\tau_{G-v, w})^2$, where the function $h$ is defined by 
        \[h(x) \coloneqq \frac{(1-\beta\gamma) x}{(\beta x + 1)(x + \gamma)}.\]
    \end{itemize}
\end{lemma}

Let $G = (V,E)$ be a graph with max degree $\Delta$.
Fix $v \in V$. Let $\mathbb{T} = \mathbb{T}_{\mathrm{SAW}}(v, G)$ and $S = \mathrm{boundary}_{\mathbb{T}, \delta}(v, N)$.
Let $G_{\mathrm{SF}} = G_{\mathrm{SF}}(G, v, S) = (V_{\mathrm{SF}}, E_{\mathrm{SF}})$ be the SAW tree with flower , and $\nu = \mu_{G_{\mathrm{SF}}}$ be the Gibbs distribution on $G_{\mathrm{SF}}$.
Let $\Lambda \subseteq V_{\mathrm{SF}}$ and $\tau \in \Omega(\nu_\Lambda)$ be a pinning such that $\nu^\tau_v(1) \in (0,1)$.
Let $\mathcal{T}$ be the induced subgraph of $G_{\mathrm{SF}}$ defined as
\[\mathcal{T} \coloneqq G_{\mathrm{SF}}[S \cup \set{\text{the connected component in $G_{\mathrm{SF}} - S$ that contains $v$}}].\]
According to the definition of $G_{\mathrm{SF}}$, we know that $\mathcal{T}$ is a tree rooted at $v$, where $S$ is a subset of leaves.
Hence \Cref{lem:formula-influence} applies to every vertices on $\mathcal{T}$.
For every vertex $u \in \mathcal{T}$, we use the notation $R_u \coloneqq R_{G(\mathcal{T}_u), u}^\tau$, where $G(\mathcal{T}_u)$ is the graph obtained by considering the sub-tree $\mathcal{T}_u$ rooted at $u$ and replace its leaves $w \in S$ with graph $G_w$ as in \Cref{lem:SAW-tree-with-flower}.

According to the tree recursion, the ratio $R_u$ can be calculated recursively by 
\[R_u \coloneqq F_{d_u}(R_{w_1}, \cdots, R_{w_{d_u}}),\] 
where $\set{w_1, \cdots, w_{d_u}}$ are children of $u$.
We first consider the following base cases:
\begin{itemize}
    \item if $u$ is pinned to $0$ or $1$, $R_u$ gets $0$ or $+\infty$;
    \item if $u \not\in S$ is free and it is a leaf of $\mathcal{T}$, we set $R_u = \lambda$;
    \item if $u \in S$, we set $R_u = R^\tau_{G_u, u}$, where $G_u$ is the graph defined in \Cref{lem:SAW-tree-with-flower}.
\end{itemize}
We note that, according to the tree recursion, for every free vertex $u$ in $\mathcal{T}$, it holds that 
\begin{align*}
    R_u \in J_{d_u}, \quad \text{where } J_{d} \coloneqq [\lambda \beta^d, \lambda\gamma^{-d}],
\end{align*}
and $d_u \coloneqq \Delta_u - 1$ such that $\Delta_u$ is the degree of the corresponding vertex in the original graph. 
Then, the squared-influence from root $v$ to a specific vertex on $\mathcal{T}$ is  the same as the squared-influence on the SAW tree by the second bullet of \Cref{lem:formula-influence}.

In~\cite{li2013correlation}, they define the following univariate function
\begin{align*}
  \alpha_d(x) \coloneqq d\sqrt{h(F_d(\*1 x)) h(x)}.
\end{align*}
\begin{lemma}[\text{\cite[Lemma 14]{li2013correlation}}] \label{lem:univariate-decay}
  If $(\beta, \gamma, \lambda)$ is up-to-$\Delta$ unique with gap $\delta$, then for every $1 \leq d < \Delta$, it holds that $\alpha_d(x) \leq \sqrt{1-\delta}$ for all $x \geq 0$.
\end{lemma}
\begin{lemma}[symmetrization] \label{lem:symmetrization}
  For any vector $\*x \in \mathbb{R}_{\geq 0}^d$, there exists $\ol{x} \geq 0$ that  
  \begin{align*}
    \sum_{i=1}^d h(x_i) \leq d \cdot h(\ol{x}) \quad \text{and} \quad F_d(\*x) = F_d(\*1 \ol{x}).
  \end{align*}
\end{lemma}
The proof of \Cref{lem:symmetrization} follows the same high-level plan as the symmetrization approach used in the previous work~\cite{li2013correlation}. We defer its proof to \Cref{sec:symmetrization}.

\begin{lemma}[\text{\cite[Lemma 36]{chen2020rapid}}]\label{lem:boundedness}
  Suppose  $(\beta,\gamma,\lambda)$ is up-to-$\Delta$ unique. For every $d \geq 1$ and for every $R_1 \in J_{d_1}$, $R_2 \in J_{d_2}$ with $1 \leq d_1,d_2 < \Delta$, $h(R_1)h(R_2) \leq C$ where $C$ is a universal constant.
\end{lemma}

Now, we are ready to prove \Cref{lem:total-squared-influence-2-spin}.
\begin{proof}[Proof of \Cref{lem:total-squared-influence-2-spin}]
  We will use potential method with potential function $h(x)$.
  For every $u \in \mathcal{T}$, let $N_u$ denote the current value of $N$ when the algorithm $\mathrm{boundary}_{\mathbb{T},\delta}(u, N)$ in \Cref{algo:boundary-on-tree} reaches vertex $u$.
  We will show by induction on $\mathcal{T}$ that for every $u \in \mathcal{T}$, it holds that 
  \begin{align} \label{eq:except-root}
    h(R_u) \cdot \infsq_{\mathcal{T}_u}(u \to S) \leq N_u^{-1} \cdot \max_{w \in S \cap \mathcal{T}_u} h(R_w),
  \end{align}
  where $\+T_u$ is the subtree rooted at $u$ and we use the notion $\inf_{\+T_u}^2(u \to S)$ as the total squared influence from $u$ to vertices in both $\+T_u$ and $S$.
  The base case, where $u$ is a leaf of $\mathcal{T}$, follows as
  \begin{align*}
  \infsq_{\mathcal{T}_u}(u \to S) = \begin{cases}
      1, & u \in S, \\
      0, & u \not\in S.
  \end{cases}
  \end{align*}
  For the inductive case where $u$ is not a leaf, let $\set{w_1, \cdots, w_{d_u}}$ be the children of $u$. We have
  \begin{align*}
    h(R_u) \infsq_{\mathcal{T}_u}(u \to S)
    &= h(R_u) \sum_{i=1}^{d_u} \infsq_{\mathcal{T}_u}(u \to w_i) \cdot \infsq_{\mathcal{T}_{w_{i}}}(w_i \to S) \\
    &= h(R_u) \sum_{i=1}^{d_u} h(R_{w_i})^2 \cdot \infsq_{\mathcal{T}_{w_{i}}}(w_i \to S) \tag{\Cref{lem:formula-influence}} \\
    &= h(R_u) \sum_{i=1}^{d_u} h(R_{w_i}) \cdot h(R_{w_i}) \infsq_{\mathcal{T}_{w_{i}}}(w_i \to S) \\
    \tag{Induction Hypothesis} &\leq h(R_u) \sum_{i=1}^{d_u} h(R_{w_i}) \cdot  \tp{N_{w_i}^{-1} \max_{w \in S\cap \mathcal{T}_{w_i}} h(R_w) } \\
    \tag{by definition of $N_{w_i}$} &\leq h(R_u) \sum_{i=1}^{d_u} h(R_{w_i}) \cdot  \tp{\frac{d_u}{1-\delta} \cdot N_{u}^{-1} \max_{w \in S\cap \mathcal{T}_{w_i}} h(R_w) } \\
    \tag{\Cref{lem:symmetrization}}
    &\leq d_u \cdot h(F_d(\*1\ol{R_{w_\cdot}})) \cdot h(\ol{R_{w_\cdot}}) \cdot  \tp{\frac{d_u}{1-\delta} \cdot N_{u}^{-1} \max_{w \in S\cap \mathcal{T}_{u}} h(R_w)} \\
    &= \frac{\alpha\tp{\*1\ol{R_{w_\cdot}}}^2}{d_u} \cdot  \tp{\frac{d_u}{1-\delta} \cdot N_{u}^{-1} \max_{w \in S\cap \mathcal{T}_{u}} h(R_w)} \\
    \tag{by \Cref{lem:univariate-decay}} 
    &\leq N_{u}^{-1} \max_{w \in S\cap \mathcal{T}_{u}} h(R_w).
  \end{align*}
  We finish the proof by applying a similar argument to the root $v$: let $\set{w_1, \cdots, w_{\Delta_v}}$ be its children,
  \begin{align*}
    \infsq_{\mathcal{T}}(v \to S)
    &= \sum_{i=1}^{\Delta_v} h(R_{w_i}) \cdot h(R_{w_i}) \infsq_{\mathcal{T}_{w_i}}(w_i \to S) \\
    \tag{by \eqref{eq:except-root}} 
    &\leq \sum_{i=1}^{\Delta_v} h(R_{w_i}) \cdot \tp{ N_{w_i}^{-1} \cdot \max_{w\in S\cap \mathcal{T}_{w_i}} h(R_{w}) } \\
    &= \frac{\Delta_v}{1-\delta} N_v^{-1} \cdot \sum_{i=1}^{\Delta_v} h(R_{w_i}) \max_{w\in S\cap \mathcal{T}_{w_i}} h(R_{w}) \\
    &\leq C_{\Delta, \delta} \cdot N_v^{-1}
  \end{align*}
  where the last inequality holds by \Cref{lem:boundedness}.
\end{proof}

With few modifications, the same proof also implies bounded total-squared-influence on trees, which may be of independent interest.
\begin{lemma} \label{lem:total-squared-influence-2-spin-on-tree}
    Let $\delta > 0$ and $\Delta \geq 1$ be an constant integer.
    Fix $\beta \geq 0, \gamma > 0, \lambda > 0$, $\beta\gamma < 1$ such that $(\beta, \gamma, \lambda)$ are up-to-$\Delta$ unique with gap $\delta$.
    Let $\mathbb{T} = (V,E)$ be a tree with max degree $\Delta$, $v$ be a vertex in $\mathbb{T}$, and $\mu = \mu_{\mathbb{T}}$ be the Gibbs distribution.
    Let $\Lambda \subseteq V$ and $\tau \in \Omega(\mu_\Lambda)$ such that $\mu^\tau_v(1) \in (0,1)$.
    For $N \geq 1$,
    \begin{align*}
        \infsq_{\mu^\tau}(v \to \mathrm{boundary}_{\mathbb{T}, \delta}(v, N)) \leq C_{\Delta,\delta} \cdot N^{-1},
    \end{align*}
     where $C_{\Delta,\delta}$ is a constant that only depends on $\Delta$ and $\delta$.
\end{lemma}

\section{Aggregate perfect marginal sampler}
\subsection{Aggregate perfect marginal sampler for explicit perfect marginal sampler}
\label{subsec:aggregate-explicit}


In this section, we describe the following abstract perfect marginal sampler for distributions over $[q]^V$ in the spirit of the framework of Anand and Jerrum. To sample the spin of vertex $u$, the sampler first identifies a set $S$ of vertices to be revealed and then determines their spins via recursive calls.

\begin{algorithm}[H]
    \KwIn{Graph $G = (V, E)$, vertex $u \in V$, and the pinning $\tau$}
    \KwOut{A marginal sample $\sigma_u \in [q]$}
    \caption{Abstract marginal sampler: MS-abstract($G,u,\tau$)}
    \label{alg:marginal-sampler-general}
    \BlankLine
    Sample $S \subseteq V$ according to some distribution $\nu = \nu(G,u,\tau)$\;\label{line:contraction-set-general}
    \ForEach{$v  \in S$}{
        $c \gets \mathrm{MS\text{-}abstract}(G,v,\tau)$\;
        Compute $\+C=\+C(G,u,v,\tau) \subseteq [q]$ of size at most $q_{\mathrm{c}}$ in $\widetilde{O}(1)$ time\;
        Update $\tau$ with $\tau \cup (v \gets c)$\; 
        \If(\tcp*[f]{Early exit}){$c \notin \+C$}{
            \Return a random sample $\sigma_u$ based on $\tau$ in $\widetilde{O}(1)$ time\;
        }
    }
    \Return a random sample $\sigma_u$ based on $\tau$ in $\widetilde{O}(1)$ time\;
\end{algorithm}

\begin{remark}
Anand and Jerrum's perfect marginal sampler for the hardcore model is easily fit into the absract framework: 
\begin{itemize}
\item The distribution $S \sim \nu$ is defined as follows: if $u$ is pinned in $\Lambda$, then $S \gets \emptyset$. Otherwise, $S$ is set to $N(u)$ with probability $\frac{\lambda}{1+\lambda}$, and to $\emptyset$ with the remaining probability;
\item The continuation set $\+C = \{0\}$, and if the spin $c \gets \text{MS-abstract}(G,v,\tau)$ is not in $\+C$, the algorithm returns $0$;
\item If the algorithm reaches the final step, the algorithm returns $1$.
\end{itemize}
Furthermore, this abstraction subsumes the perfect marginal samplers for spin systems with weak interactions \cite{LWY26} and polymer models. We provide a detailed discussion of these connections later in this section.
\end{remark}

As mentioned in the introduction, the underlying subcritical branching process of Anand and Jerrum's perfect marginal sampler for the hardcore model enables the simultation of $N$ independent executions in a single batch.
This intuition carries over to the abstract marginal sampler, which we formalized as below.

\begin{theorem}\label{thm:general-explicit}
     Let $\delta \in (0,1)$, $D \ge 1$ and $q \ge 2$ be constants. Suppose the distribution $\nu$ in \Cref{line:contraction-set-general} of \Cref{alg:marginal-sampler-general} satisfies $\mathbb{E}_{\nu}[|S|] \le 1 - \delta$ and $\abs{\bigcup_{S \in \mathrm{supp}(\nu)} S} \le D$. 
     
    There exists an algorithm that, given a graph $G=(V,E)$, a vertex $u \in V$, computes the list of indicator sums $\tp{\sum_{i=1}^N \mathbf{1}[X_i = c]}_{c \in [q]}$ in $\widetilde{O}(N^{1-\alpha})$ time, where $X_1, \dots, X_N$ are independent marginal samples of vertex $u$ generated by {MS-abstract}($G,u,\emptyset$), and $\alpha = \frac{\delta}{100 D \log q}$.
\end{theorem}

Analogue to the perfect marginal sampler for the hardcore model, we consider the following aggregate version of~\Cref{alg:marginal-sampler-general}. 
The aggregate perfect marginal sampler draws a sample $(Q_S)_{S}$ from multinomial distribution $M(N,\nu)$ to simulate the execution of~\Cref{line:contraction-set-general} in~\Cref{alg:marginal-sampler-general} across $N$ independent runs of the perfect marginal samplers.
For each subset $S$, the procedure would simulate $Q_S$ independent runs in a single batch by dividing $Q_S$ queries into several branches depending on the outputs of recursive calls. Hence, the aggregate perfect marginal sampler indeed faithfully simulates $N$ independent runs of perfect marginal sampler in~\Cref{alg:marginal-sampler-general}. It remains to show the efficiency of the aggregate perfect marginal sampler, which is given as follows.

\begin{algorithm}[H]
    \KwIn{Graph $G = (V, E)$, vertex $u \in V$, pinning $\tau$, and a non-negative integer $N$ denoting the number of samples}
    \KwOut{A non-negative integer list of size $q$, denoting the number of samples for each spin}
    \caption{Aggregate abstract marginal sampler: aggregate-MS-abstract($G,u,\tau,N$)}
    \label{alg:aggregate-marginal-sampler-general}
    \BlankLine
    \If{$N = 0$}{
        \Return $[0,0,\ldots,0]$\;
    }
    Sample an integer list $(Q_{S})_{S}$ according to the multinomial distribution $M(N;\nu)$, where $\nu$ is given in~\Cref{line:contraction-set-general} of~\Cref{alg:marginal-sampler-general}\;\label{line:contraction-set-general-aggregate}
    Initialize an integer list $R \gets [0,0,\ldots,0]$ of size $q$\;
    \tcp{The list $R$ tracks the number of samples for each spin}
    
    \ForEach{$S \in \mathrm{supp}(\nu)$\label{line:enumerate}}{
        $\text{query-list} \gets [(\tau,Q_S)]$ and $\text{new-query-list} \gets []$\;
        \ForEach{$v \in S$}{
            \ForEach{$(\tau^\star,Q^\star)$ in query-list}{
                Compute $\+C=\+C(G,u,v,\tau^\star) \subseteq [q]$ of size at most $q_{c}$\;
                $[r_1,r_2,\ldots,r_q] \gets \mathrm{aggregate\text{-}MS\text{-}abstract}(G,v,\tau^\star,Q^\star)$\;
                \ForEach(\tcp*[f]{Continuing branches}){$c \in \+C$}{
                    Append $(\tau^\star \cup (v \gets c),r_c)$ to the end of new-query-list\;
                }
                \ForEach(\tcp*[f]{Branches early terminated}){$c \not\in \+C$}{
                    Compute the frequency list $[r_1,r_2,\ldots,r_q]$ based on $\tau^\star \cup (v \gets c)$\;
                    Set $R \gets R+[r_1,r_2,\ldots,r_q]$\;
                }
            }
            $\text{query-list} \gets \text{new-query-list}$ and set $\text{new-query-list} \gets []$\;
        }
        \ForEach{$(\tau^\star,Q^\star)$ in query-list}{
            Compute the frequency list $[r_1,r_2,\ldots,r_q]$ based on $\tau^\star$\;
            Set $R \gets R+[r_1,r_2,\ldots,r_q]$\;
        }
    }
    \Return $R$\;
\end{algorithm}

\begin{remark}
    If $\nu$ has large support, we represent the sample in \Cref{line:contraction-set-general-aggregate} as a sparse mapping $\{(S, Q_S)\}$ and enumerate its entries accordingly in \Cref{line:enumerate}.
\end{remark}


\begin{lemma}\label{lem:number-of-calls}
    For $\alpha \in (0,1)$ and $k \ge 0$ be an integer, define
    \begin{align*}
        g_\alpha(k) \coloneqq \sum_{i=0}^{k-1} q_{c}^{\alpha i},
    \end{align*}
    where the parameter $q_c$ is the upper bound for the size of continuation set $\+C$ in~\Cref{alg:marginal-sampler-general}.
    Let
    \begin{align*}
        \beta(\alpha) \coloneqq \sup \sum_{S} \nu(S)^{1-\alpha} \cdot g_\alpha(\abs{S}),
    \end{align*}
    where the supremum is taken over all scenarios (in particular, over all distributions $\nu$).
    Suppose there exists $\alpha \in (0,1)$ such that
    $\beta(\alpha) < 1$ and~\Cref{line:contraction-set-general-aggregate} can be implemented in $\widetilde{O}(N^{1-\alpha})$ time in expectation, then the aggregate perfect marginal sampler in~\Cref{alg:aggregate-marginal-sampler-general}
    can be implemented in $\widetilde{O}\tp{N^{1-\alpha}}$ time in expectation.
\end{lemma}

\begin{proof}
    Without loss of generality, we assume $q_{c} = q$ and~\Cref{line:contraction-set-general-aggregate} takes exactly $N^{1-\alpha}$ time to process.
    Let $E_{N}$ be an upper bound for the total cost (of time) to generate $N$ samples in~\Cref{alg:aggregate-marginal-sampler-general} in expectation.
    We will show that $E_N = C \cdot N^{1-\alpha}$ (when $N \ge 1$) is indeed an upper bound for some constant $C \ge 1$ to be determined.
    Suppose a recursive call aggregate-MS-abstract$(\cdot,\cdot,\cdot,m)$ incurs a cost of $E_m$.
    To prove the claim, it suffices to show that the total expected cost of aggregate-MS-abstract$(\cdot,\cdot,\cdot,N)$ is bounded by $E_N$.

    Suppose we have drawn the integer list $(Q_S)$ in~\Cref{line:contraction-set-general-aggregate}. Fix a subset $S$ and condition on the event
    that $Q_S=m$. We first show that the process of handling this subset incurs a total cost at most
    \begin{align}
        F_{\abs{S},m} \le C \cdot m^{1-\alpha} \cdot g_\alpha(\abs{S}).
        \label{eq:F-bound}
    \end{align}

    We prove~\eqref{eq:F-bound} by induction on $\abs{S}$. The base case $\abs{S}=0$ is immediate. Let $\abs{S}=k\ge 1$ and consider the first vertex $v\in S$ processed by the algorithm.
    To handle this subset, the algorithm makes one recursive call
    $\mathrm{aggregate\text{-}MS\text{-}abstract}(G,v,\tau,m)$ and obtains a list $[r_1,\ldots,r_q]$ with $\sum_{c=1}^q r_c = m$. Then, the algorithm proceeds with $q$ subsets of size $k-1$ with parameters $r_1,r_2,\ldots,r_q$ to compute the correct distribution for set $S$. By induction hypothesis, the total cost is bounded by
    \begin{align}\label{eq:cost}
        F_{k,m} \le C \cdot m^{1-\alpha} + \sum_{i=1}^q F_{k-1,r_i} \le C \cdot m^{1-\alpha} + C \cdot g_\alpha(k-1) \cdot \sum_{i=1}^q r_i^{1-\alpha}. 
    \end{align}
    Since $\alpha\in(0,1)$, the map $x\mapsto x^{1-\alpha}$ is concave, Jensen's inequality yields
    \begin{align*}
        \sum_{c=1}^q r_c^{1-\alpha}
        \le q\cdot \tp{\frac{\sum_{c=1}^q r_c}{q}}^{1-\alpha}
        = q\cdot \tp{\frac{m}{q}}^{1-\alpha}
        = q^\alpha \cdot m^{1-\alpha}.
    \end{align*}
    Plugging into~\eqref{eq:cost}, this proves~\eqref{eq:F-bound}.

    Since the multinomial distribution $M(N;\nu)$ in \Cref{line:contraction-set-general-aggregate} can be implemented in $N^{1-\alpha}$ cost in expectation.
    Plugging the above bound into the total cost, we obtain
    \begin{align*}
        E_N
        &\le N^{1-\alpha} + \E[(Q_S)\sim M(N;\nu)]{\sum_{S} F_{\abs{S},Q_S}}\\
        &\le N^{1-\alpha} + C \cdot \E[(Q_S)\sim M(N;\nu)]{\sum_{S} g_\alpha(\abs{S}) \cdot (Q_S)^{1-\alpha}}.
    \end{align*}
    For each fixed $S$, $Q_S$ has marginal distribution $\mathrm{Bin}(N,\nu(S))$, and by concavity of $x^{1-\alpha}$ again,
    \begin{align*}
        \E[]{Q_S^{1-\alpha}} \le \tp{\E[]{Q_S}}^{1-\alpha}= \tp{N\nu(S)}^{1-\alpha}.
    \end{align*}
    Hence,
    \begin{align*}
        E_N
        &\le N^{1-\alpha} + C \cdot \sum_{S} g_\alpha(\abs{S}) \cdot (N\nu(S))^{1-\alpha}\\
        &= N^{1-\alpha} + C \cdot N^{1-\alpha} \cdot \sum_{S} \nu(S)^{1-\alpha} \cdot g_\alpha(\abs{S})\\
        &\le N^{1-\alpha} + C \cdot \beta(\alpha) \cdot N^{1-\alpha},
    \end{align*}
    where we recall $\beta(\alpha) = \sup \sum_{S} \nu(S)^{1-\alpha} g_\alpha(\abs{S})$ and the supremum is taken over all distribution $\nu$.

    Finally, choose $C \coloneqq \frac{1}{1-\beta(\alpha)}$. 
    Since $\beta(\alpha) < 1$, $C$ is positive.
    Then for all $N\ge 1$,
    \begin{align*}
        N^{1-\alpha} + C \beta(\alpha) \cdot N^{1-\alpha}
        = C(1-\beta(\alpha)) \cdot N^{1-\alpha} + C \beta(\alpha) \cdot N^{1-\alpha}
        = C \cdot N^{1-\alpha},
    \end{align*}
    and therefore $E_N \le C \cdot N^{1-\alpha}$, completing the proof.
\end{proof}

Observe that $\beta(\alpha)$ approaches $\sup \E[\nu]{\abs{S}}$ as $\alpha$ approaches $0$. 
Thus, for perfect marginal samplers exhibiting a subcritical branching process, a sufficiently small $\alpha > 0$ is guaranteed to exist. 
The following lemma provides a quantitative bound for this parameter.

\begin{lemma}\label{lem:exponent-1}
   Under the assumptions in~\Cref{thm:general-explicit}, it holds $\beta(\alpha) < 1$ for $\alpha = \frac{\delta}{100 D \log q}$.
\end{lemma}

\begin{proof}
    Let $\theta = 1/100$ for ease of notation.
    Fix a distribution $\nu$ in the supremum of $\beta(\alpha)$, it holds 
    \begin{align*}
        g_\alpha(k) \le q^{\alpha D} \cdot k = \exp\tp{\theta \delta} \cdot k.
    \end{align*} 
    Let $P_k = \Pr[S \sim \nu]{\abs{S} = k}$ and $U = \bigcup_{S \in \mathrm{supp}(\nu)} S$. 
    Recall that $\abs{U} \leq D$.
    By the concavity of $x \mapsto x^{1-\alpha}$, 
    \begin{align*}
        \sum_{\substack{S \subseteq U\\\abs{S} = k}} \nu(S)^{1-\alpha} \le \binom{D}{k}^{\alpha} P_k^{1-\alpha} \le 2^{\alpha D} \cdot P_k^{1- \alpha} \le \exp\tp{\theta \delta} \cdot P_k^{1-\alpha} 
    \end{align*}
    Therefore, $\beta(\alpha)$ is bounded by
    \begin{align*}
        \beta(\alpha) \le \exp\tp{2 \theta \delta}\sum_{k=1}^{D} k \cdot P_k^{1-\alpha} &\le \exp\tp{2\theta \delta} \cdot D^\alpha \sum_{k=1}^{D} (k \cdot P_k)^{1-\alpha} \\ 
        &\le \exp\tp{2\theta\delta} D^{2 \alpha } (1-\delta)^{1-\alpha} \le \exp(4\theta \delta) (1-\delta)^{1-\theta} < 1. \qedhere
    \end{align*}
\end{proof}

\Cref{thm:general-explicit} immediately follows from~\Cref{lem:number-of-calls} and~\Cref{lem:exponent-1}. 
Though, for specific applications, the exponent $\alpha$ can be further improved. 

\paragraph{Hardcore model}
    The original Anand and Jerrum's perfect marginal sampler in~\Cref{alg:AJ-marginal-sampler-new} is the prototype of the generalized version in~\Cref{alg:marginal-sampler-general}: The distribution $\nu$ either samples $N(u)$ with probability $\frac{\lambda}{1+\lambda}$, or an empty set with the remaining probability $\frac{1}{1+\lambda}$. The continuation set $\+C$ contains the single element $0$, thus $q_c = 1$. 
    In this case, the exponent $\alpha$ can be improved from $\Omega\tp{\frac{1}{\Delta}}$ to $\Omega \tp{\frac{1}{\log \Delta}}$ as shown in the following proposition.
    \begin{proposition}\label{prop:hardcore}
        Consider the hardcore model specified by a graph $G$ with maximum degree $\Delta$ and fugacity $\lambda$ satisfying $\frac{\lambda}{1+\lambda} = \frac{1}{c\Delta}$ for some $c > 1$. The exponent $\alpha < \frac{\log c}{\log \Delta + \log c}$ satisfies
        \begin{align*}
            \beta(\alpha) = \Delta \cdot \tp{\frac{\lambda}{1+\lambda}}^{1-\alpha} < 1.
        \end{align*}
    \end{proposition}
    The proof of~\Cref{prop:hardcore} follows from a direct calculation. As a corollary of~\Cref{lem:number-of-calls,prop:hardcore}, for the hardcore model satisfying the Dobrushin condition (i.e., $\lambda < \frac{1}{\Delta-1}$), we can compute the empirical mean of $N = O\tp{n/\epsilon^2}$ samples in sublinear time, thus implying a subquadratic time approximating algorithm for the partition function.

    \begin{corollary}[Hardcore model]\label{corollary:hardcore}
    Let $\Delta \ge 3$ be a constant, and $\lambda \le \frac{1}{c \Delta-1}$ for some $c > 1$. For graph with maximum degree $\Delta$ and $n$ vertices, there exists an $\mathsf{FPRAS}$ for the partition function of the hardcore model with fugacity $\lambda$ in time $\widetilde{O}(n^{1+\frac{\log \Delta}{\log \Delta + \log c} + \delta} \cdot \epsilon^{-\frac{2 \log \Delta}{\log \Delta + \log c} + 2\delta})$ for any constant $\delta > 0$, where $\epsilon > 0$ is the (multiplicative) error bound and $n$ is the number of vertices in the graph.
\end{corollary}

\begin{remark}
    Analyzing carefullly, the small slackness $\delta > 0$ can be fully removed.
\end{remark}
\begin{remark}
    When $\lambda \le \frac{1}{\Delta^{1.5+\zeta}}$ for some small $\zeta > 0$ considered in~\cite{AFFGW25}, the runtime of the aggregate sampler becomes $\widetilde{O}(n^{1.5})$, which improves their $\widetilde{O}(n^{2-\delta})$ result for some small $\delta > 0$.
\end{remark}

\paragraph{Polymer model}

Consider a polymer model $(\+C(\cdot),w,\+G)$ with $q$-spins on a class of $\+G$ of graphs with maximum degree $\Delta$ satisfies the \emph{polymer sampling condition} up to a constant, i.e. there exist $C \ge 10$ and $\theta \ge C\tp{1+ \log \tp{(q-1)\Delta}}$ such that
\begin{align}\label{eq:polymer-condition}
    \forall G \in \+G \text{ and } \gamma \in \+C(G), \quad w_\gamma \le \exp\tp{-\theta \abs{\gamma}}.
\end{align}
Inspired by Anand and Jerrum's perfect marginal sampler and the polymer chain~\cite{CGGPSV21}, we consider the following perfect marginal sampler for the polymer model. Given an undirected graph $G = (V,E)$, for each vertex $u \in V$, we define the distribution $\nu_u$ as follows: For each non-empty polymer $\gamma$ including $u$, sample $\gamma \sim \nu_u$ with probability $w_\gamma$; with the remaining probability, return the empty set.
With the oracle $\nu$, we are able to introduce the perfect marginal sampler.
For convenience, we  treat the connected component $\gamma$ as a vertex set, and use the following notation:
\begin{align*}
    N^+(\gamma) = \bigcup_{u \in \gamma} N^+(u).
\end{align*}

\begin{algorithm}[H]
    \KwIn{Graph $G=(V, E)$, vertex $u \in V$, pinned vertices $\Lambda \subseteq V$ and ground-state spins $g \in [q]^V$;}
    \KwOut{Either an empty set or a polymer $\gamma$ including $u$, following the marginal distribution on vertex $u$ conditioned on vertices in $\Lambda$ being pinned by their ground-state spins}
    \KwOracle{Distributions $\nu_v$ for each $v \in V$}
    \caption{Marginal sampler for polymer model: MS-polymer($G,u,\Lambda,g$)}
    \label{alg:polymer-marginal-sampler}   
    \BlankLine
    \If{$u \in \Lambda$}{
        \Return $\emptyset$\;
    }
    Propose a polymer $\gamma \sim \nu_u$\;
    
    \If{$\gamma \cap \Lambda \neq \emptyset$}{
        \Return $\emptyset$\;
    }

    \ForEach(\tcp*[f]{Enumerate vertex $u$ first}){$v \in N^+(\gamma)$\label{line:enumerate-first}}{ 
        $\gamma' \gets \text{MS-polymer}(G,v,\Lambda,g)$\;
        \If{$ \gamma'\neq \emptyset$}{
            \Return $\emptyset$\;
        }
        $\Lambda \gets \Lambda \cup \{v\}$\;
    }
    \Return $\gamma$\;
\end{algorithm}

\begin{remark}
    In~\Cref{line:enumerate-first}, we enumerate the vertex $u$ first. The condition is optional but serves to simplify the analysis of correctness.
\end{remark}


Let $\Gamma \subseteq \mathcal{C}(G)$ and $u$ be a vertex in $G$.
We will use the following notation
\begin{align} \label{eq:def-Gamma-u}
    \Gamma(u) \coloneqq \begin{cases}
        \gamma, &\text{if $\gamma \in \Gamma$ contains $u$,} \\
        \emptyset, &\text{no such $\gamma \in \Gamma$.}
    \end{cases}
\end{align}
Also, we will use the following notation to denote the fact that no $\gamma \in \Gamma$ contains any vertex in $\Lambda$:
\begin{align*}
    \Gamma \cap \Lambda = \emptyset.
\end{align*}

The following lemma claims the correctness of the perfect marginal sampler.
\begin{lemma}
    The perfect marginal sampler~\Cref{alg:polymer-marginal-sampler} terminates in finite time almost surely and 
    \begin{align}\label{eq:polymer-correctness-criterion}
       \Pr[]{\text{MS-polymer}(G,u,\Lambda,g) \text{ outputs } \gamma} 
       &= \Pr[\Gamma \sim \mu]{\Gamma(u) = \gamma \mid \Gamma \cap \Lambda = \emptyset}.
    \end{align}
\end{lemma}

\begin{proof}
    It suffices to show the following properties of the perfect marginal sampler:
    \begin{enumerate}
    \item The process would terminate with probability $1$;
    \item Suppose each recursive call $\text{MS-polymer}(G,v,\Lambda,g)$ in~\Cref{alg:polymer-marginal-sampler} returns a polymer with the correct marginal distribution for vertex $v$, then the whole process would return a polymer with the correct marginal distribution for vertex $u$.
    \end{enumerate}
    
    The first property follows from the polymer sampling condition: The expected number of recursive calls is bounded by
    \begin{align}\label{eq:calc-polymer}
        \nonumber \sum_{\gamma \text{ contains } u} w_\gamma \cdot (\Delta + 1)\abs{\gamma} &\le \sum_{k=1}^{+\infty} (\Delta + 1) k \cdot \sum_{\substack{\gamma \text{ contains } u\\\abs{\gamma}=k}} w_\gamma \\
        \nonumber &\overset{(\star)}{\le} \sum_{k \ge 1}(\e \Delta)^{k-1} (q-1)^{k} (\Delta+1) k \cdot \exp\tp{-\tau k}\\
        \nonumber &\le \frac{\Delta+1}{\e \Delta}\sum_{k \ge 1} k \cdot \exp\tp{-(\tau-1-\log (q-1)\Delta) k}\\
        &< 1,
    \end{align}
    where the inequality $(\star)$ follows from the polymer sampling condition and~\Cref{lem:connected-subgraph}.
    
    Now it remains to check the second property. Without loss of generality, we assume that $u \not\in \Lambda$. 
    We will prove \eqref{eq:polymer-correctness-criterion} by induction on the number of free vertices $\abs{V\setminus\Lambda}$.
    For the base case where $V \setminus \Lambda = \set{u}$, let $p_\gamma \coloneqq \Pr[]{\text{MS-polymer}(G,u,\Lambda,g) \text{ outputs }\gamma}$ for $\gamma \in \set{\emptyset, \set{u}}$, we have
    \begin{align*}
        p_\emptyset = (1 - w_{\set{u}}) + w_{\set{u}} \tp{1 - p_\emptyset}.
    \end{align*}
    This implies that $p_\emptyset = \frac{1}{1+w_{\set{u}}} = 1 - p_{\set{u}}$, and finishes the proof for the base case.

    For the inductive case, we also use the notation $p_\gamma \coloneqq \Pr[]{\text{MS-polymer}(G,u,\Lambda,g) \text{ outputs }\gamma}$ for all $\gamma$.
    Note that in the for-loop of in line 6 of \Cref{alg:polymer-marginal-sampler}, we can always handle vertex $u$ first.
    Then, according to the induction hypothesis, for any $\gamma \ni u$, we have
    \begin{align*}
        p_\gamma 
        &= w_\gamma \cdot p_\emptyset \cdot \Pr[\Gamma\sim \mu]{\Gamma \cap (\Lambda \cup N^+(\gamma)) = \emptyset \mid \Gamma \cap (\Lambda \cup \set{u}) = \emptyset} \\
        &= w_\gamma \cdot p_\emptyset \cdot  \frac{\sum_{\Gamma \cap (\Lambda \cup N^+(\gamma) = \emptyset} \prod_{\gamma' \in \Gamma} w_{\gamma'}}{\sum_{\Gamma \cap (\Lambda \cup \set{u}) = \emptyset} \prod_{\gamma' \in \Gamma} w_{\gamma'}} \\
        &= p_\emptyset \cdot \frac{\sum_{\Gamma \cap \Lambda = \emptyset \text{ and } \Gamma(u) = \gamma} \prod_{\gamma' \in \Gamma} w_{\gamma'}}{\sum_{\Gamma \cap \Lambda = \emptyset \text{ and } \Gamma(u) = \emptyset} \prod_{\gamma' \in \Gamma} w_{\gamma'}},
    \end{align*}
    where we recall \eqref{eq:def-Gamma-u} for the definition of $\Gamma(u)$.
    Since we also have $p_\emptyset + \sum_{\gamma \ni u} p_\gamma = 1$, it gives
    \begin{align*} 
        p_\emptyset &=  \frac{\sum_{\Gamma \cap \Lambda = \emptyset \text{ and } \Gamma(u) = \emptyset } \prod_{\gamma' \in \Gamma} w_{\gamma'}}{\sum_{\Gamma \cap \Lambda = \emptyset} \prod_{\gamma' \in \Gamma} w_{\gamma'}} = \Pr[\Gamma\sim \mu]{\Gamma(u) = \emptyset \mid \Gamma\cap \Lambda = \emptyset} \\
        \forall \gamma \ni u, \quad p_\gamma
        &= \frac{\sum_{\Gamma \cap \Lambda = \emptyset \text{ and } \Gamma(u) = \gamma} \prod_{\gamma' \in \Gamma} w_{\gamma'}}{\sum_{\Gamma \cap \Lambda = \emptyset} \prod_{\gamma' \in \Gamma} w_{\gamma'}}
        =  \Pr[\Gamma\sim \mu]{\Gamma(u) = \gamma \mid \Gamma\cap \Lambda = \emptyset}.
    \end{align*}
    This finishes the proof.
\end{proof}

To fit into the framework of the abstract marginal sampler (\Cref{alg:marginal-sampler-general}), we slightly modify the marginal sampler to return a spin of a given vertex.
The correctness of the modified algorithm follows directly from the definition of the polymer model.

\begin{algorithm}[H]
    \KwIn{Graph $G=(V, E)$, vertex $u \in V$, pinning $\tau$ and ground-state spins $g \in [q]^V$;}
    \KwOut{A marginal sampler $\sigma_u \in [q]$}
    \KwOracle{Distributions $\nu_v$ for each $v \in V$}
    \caption{Modified marginal sampler for polymer model: MMS-polymer($G,u,\tau,g$)}
    \label{alg:polymer-marginal-sampler-2}   
    \BlankLine
    \If{$u \in \mathrm{supp}(\tau)$}{
        \Return $\sigma_u$\;
    }
    Propose a polymer $\gamma \sim \nu_u$\;
    \If{$\gamma \cap \mathrm{supp}(\tau) \neq \emptyset$}{
        \Return $g_u$\;
    }
    \ForEach{$v \in N^+(\gamma)$}{
        $c \gets \text{MMS-polymer}(G,v,\tau,g)$\;
        \If{$c \neq g_v$}{
            \Return $g_u$\;
        }
        $\tau \gets \tau \cup (v \gets g_v)$\;
    }
    \Return $\gamma_u$\;
\end{algorithm}

Unfortunately, since the support size $\nu_u$ is no longer constant, \Cref{thm:general-explicit} is no longer applicable. However, the prerequisites for \Cref{lem:number-of-calls} remain satisfied: (1) a sample from the multinomial distribution $M(N;\nu_v)$ can be generated in sub-linear time, and (2) $\beta(\alpha) < 1$ for a chosen $\alpha$. This ensures the efficiency of the ``aggregated'' perfect marginal sampler for the polymer model.

\begin{lemma}[Efficiency of sampling from multinomial distribution]
    The multinomial distribution $M(N;\nu_v)$ can be implemented in $\widetilde{O}(N^{2/(C-2)})$ time in expectation.
\end{lemma}

\begin{proof}
The sampling procedure for $M(N;\nu_u)$ generalizes the implementation of distribution $\nu_u$ described in~\cite{CGGPSV21}:
\begin{enumerate}
\item Sample $\{(k,X_k): k, X_k \ge 1\} \sim M(N;\mathrm{Geo}(1-\exp\tp{-r}))$, where $r = \theta -2 -\log (q-1)\Delta$.
\item For each entry $(k,X_k)$, enumerate all non-empty polymers $\gamma$ containing $u$ of size at most $k$, and then sample from multinomial distribution $M(X_k; \widetilde{\nu}_k)$, where $\widetilde{\nu}_k(\gamma) = w_\gamma \cdot \exp\tp{r \abs{\gamma}}$ for all such non-empty polymer $\gamma$, and the distribution $\widetilde{\nu}_k$ returns empty set with the remaining probability.
\end{enumerate}

We now show that both steps can be implemented efficiently. To implement the first step, we maintain a random variable $M$ initialized by $N$. In the $t$-th step, we let $X_t \sim \mathrm{Bin}(M,1-\exp\tp{-r})$ and update $M$ with $M - X_t$. The process terminates when $M$ reaches zero. 
Let $m = \max(Z_1,Z_2,\ldots,Z_N)$ where $Z_1,Z_2,\ldots,Z_N$ are independent copies of $\mathrm{Geo}(1-\exp\tp{-r})$.
The expected number of rounds is exactly $\E[]{m}$, which is $\Theta(\log N)$. 

By~\Cref{lem:connected-subgraph,lem:binomial}, the second part can be implemented within time $\widetilde{O}\tp{\E[]{\sum_{k=1}^m k^7 (\e \Delta)^{2k}}}$, which is in the order of $\widetilde{O}\tp{\E[]{m^8 (\e \Delta)^{2m}}} $.
By a standard argument, it holds
\begin{align*}
    \E[]{m^8 (\e \Delta)^{2m}} &\le \sum_{k \ge 1} \Pr[]{m \ge k} \cdot k^8 (\e \Delta)^{2k}\\
    &= \sum_{k \ge 1} \tp{1-\tp{1 - \exp\tp{-r(k-1)}}^N} \cdot k^8 (\e \Delta)^{2k}\\
    &\le \sum_{k \ge 1} \min\tp{1,N \exp\tp{-r(k-1)}} \cdot k^8 (\e \Delta)^{2k}\\
    &= \widetilde{O}(N^{\frac{2 \log \e \Delta}{r}}).
\end{align*}
By the choice of parameter, we have  $r \ge (C-2) \log \e \Delta$, thus the expected total runtime of the second step is bounded by $\widetilde{O}(N^{\frac{2}{C-2}})$.
\end{proof}

\begin{proposition}
    The exponent $\alpha = 1-\frac{4}{C-2}$ satisfies
    \begin{align*}
        \beta(\alpha) \le \sum_{\substack{\gamma \text{ containing } u}} w_\gamma^{1-\alpha} \cdot g_\alpha((\Delta + 1) \abs{\gamma}) = \sum_{\substack{\gamma \text{ containing } u}} (\Delta+1) \abs{\gamma} \cdot w^{1-\alpha}_\gamma < 1.
    \end{align*}
\end{proposition}

We note that the equation in the middle holds because $q_c = 1$ in function $g_\alpha(\cdot)$ (i.e. $\mathcal{C} = \set{g_v}$ in \Cref{alg:marginal-sampler-general}).
The calculation is almost the same as~\eqref{eq:calc-polymer} and we omit the proof.
As a corollary, for polymer models satisfying the polymer sampling condition up to a constant factor, we have a subquadratic time approximating algorithm for the partition function via the self-reducibility.
The latter part of~\Cref{cor:polymer} follows from the fact that the hardcore model on $\rho$-expander bipartite graphs with fugacity $\lambda \ge (3\Delta)^{C/\rho}$ can be expressed as a polymer model satisfying the polymer sampling condition(see~\cite{CGGPSV21}).
\begin{corollary}[Polymer model]\label{cor:polymer}
    Let $\Delta \ge 3$ and $C \ge 10$ be constants. There exist $\mathsf{FPRAS}$ for estimating the partition function of polymer models on graphs with maximum degree $\Delta$ and $n$ vertices satisfying~\eqref{eq:polymer-condition} with constant $C$ in expected time $\widetilde{O}(n^{1+\frac{4}{C-2}} \eps^{-\frac{8}{C-2}})$, where $\epsilon > 0$ denotes the multiplicative error bound.

    Additionally, there exists $\mathsf{FPRAS}$ for estimating the partition function of the hardcore model on $\rho$-expander bipartite graphs with maximum degree $\Delta$ and fugacity $\lambda \ge (3\Delta)^{C/\rho}$ in $\widetilde{O}\tp{n^{1+\frac{4}{C-2}} \cdot \eps^{-\frac{8}{C-2}}} $ time in expectation, where $C \ge 10$ and $\rho \in (0,1)$ are constants.
\end{corollary}

\paragraph{Soft-constraint spin systems with sufficiently weak interactions}
Let $\+S = (G,\lambda,A)$ be a $q$-spin system on graph $G=(V,E)$. The spin system is said to exhibit \emph{sufficiently weak interactions} if for any pair of spins $(u,v) \in [q]^2$, it holds
\begin{align*}
    1-\frac{1}{2\Delta}\le A(u,v) \le 1.
\end{align*}
For such spin systems, we consider the following perfect marginal sampler proposed in~\cite{LWY26}.

\begin{algorithm}[H]
    \KwIn{An undirected graph $G=(V, E)$, vertex $u \in V$ and pinning $\tau \in [q]^\Lambda$ where $\Lambda \subseteq V$;}
    \KwOut{A spin $\sigma_u \in [q]$ following the distribution of $\mu_u^\tau$;}
    \caption{Marginal sampler for spin systems with weak constraints: MS-spin($G,u,\tau$)}
    \label{alg:LWY-sampler}
    
    \BlankLine
    \If{$u \in \Lambda$}{\Return $\tau_u$\;}
    $\tau_{\mathrm{new}} \gets \emptyset$\;
    \ForEach{$v \in N(u)$}{\label{line:begin}
            $X \sim \mathrm{Ber}(1/2\Delta)$\;\label{line:bin-most-LWY}
            \If{$X = 1$}{
                $c \gets \text{MS-spin}(G,v,\tau \cup \tau^{\mathrm{new}})$\;
                $\tau_{\mathrm{new}} \gets \tau_{\mathrm{new}} \cup (v \gets c)$\;\label{line:oracle}
            }    
    }
    Sample a spin $c \in [q]$ with probability proportional to $\lambda_u \in \mathbb{R}_{\ge 0}^q$\;\label{line:assign-spin}
    $\text{accept} \gets \texttt{true}$\;
    \ForEach{$v \in \Lambda_{\mathrm{new}} = \mathrm{supp}(\tau)$}{
        \If{$\mathrm{Ber}(2\Delta \cdot (1-A(c,\tau_{\mathrm{new},v})) = 1$}{
             $\text{accept} \gets \texttt{false}$\;\label{line:reject}
        }
    }
    \If{$\text{accept} = \texttt{true}$}{\label{line:accept-branch}
        \Return $c$\;
    }
    \Return MS-spin($G,u,\tau \cup \tau_{\mathrm{new}}$)\;
\end{algorithm}

To fit into the framework of the abstract marginal sampler (\Cref{alg:marginal-sampler-general}), we translate the marginal sampler into the following form.

\begin{algorithm}[H]
    \KwIn{An undirected graph $G=(V, E)$, vertex $u \in V$ and pinning $\tau \in [q]^\Lambda$ where $\Lambda \subseteq V$;}
    \KwOut{A spin $\sigma_u \in [q]$ following the distribution of $\mu_u^\tau$;}
    \caption{Modified marginal sampler for spin systems: MMS-spin($G,u,\tau$)}
    \label{alg:LWY-sampler-modified}
    
    \BlankLine
    \If{$u \in \Lambda$}{\Return $\tau_u$\;}
    Sample $S \sim \mathrm{Bin}(N(u),1/2\Delta)$\;
    \If{$S \neq \emptyset$}{
        $S \gets S \cup \{u\}$\;
    }
    \ForEach(\tcp*[f]{Enumerate $u$ last whenever $S$ is non-empty}){$v \in S$}{
        $c \gets \text{MMS-spin}(G,v,\tau)$\;
        $\tau \gets \tau \cup (v \gets c)$\;
    }
    \Return a spin $\sigma_u$ in $[q]$ according to the following rules:
    \begin{enumerate}
        \item Sample a spin $c \in [q]$ with probability proportional to $\lambda_u \in \mathbb{R}_{\ge 0}^q$;
        \item With probability $\prod_{v \in S \setminus \{u\}} (2\Delta \cdot (1-A(c,\tau_{v})))$, set $\sigma_u \gets c$; otherwise, set $\sigma_u \gets \tau_u$.
    \end{enumerate}
\end{algorithm}

The perfect marginal sampler in~\Cref{alg:LWY-sampler-modified} is easily seen to fit into the framework of abstract marginal sampler in~\Cref{alg:marginal-sampler-general}. The distribution $\nu$ is the law of subset $S \subseteq N^+(u)$ generated by adding each vertex $v \in N(u)$ into $S$ independently with probability $\frac{1}{2\Delta}$, and adding vertex $u$ into $S$ if $S$ is non-empty. The continuation set $\+C = [q]$. In this case, it holds
\begin{align*}
    \beta(\alpha) = \sum_{k=1}^{\Delta} \binom{\Delta}{k} \frac{(2\Delta-1)^{(\Delta-k) (1-\alpha)}}{\tp{2\Delta}^{\Delta (1-\alpha)}}g_\alpha(k+1),
\end{align*}
where $g_\alpha(k) = \frac{q^{\alpha k} - 1}{q^{\alpha} - 1}$.
We show that, exponent $\alpha$ could be also improved to $\Omega\tp{\frac{1}{\log q\Delta}}$. 

\begin{proposition}\label{prop:two-spin}
     Let $\Delta \ge 3$ be a constant.
    The exponent $\alpha = \frac{1}{1000 \log q\Delta}$ satisfies $\beta(\alpha) < 1$.
\end{proposition}

\begin{proof}[Proof of~\Cref{prop:two-spin}]
By a straightforward calculation, $g_\alpha(k) \le q^{\alpha (k-1)} \cdot k$. 
Therefore, $\beta(\alpha)$ can be bounded by
\begin{align*}
    \beta(\alpha) \le \sum_{k=1}^{\Delta} \binom{\Delta}{k} \frac{(2\Delta-1)^{(\Delta-k) (1-\alpha)}}{(2\Delta)^{\Delta (1-\alpha)}} (k+1) q^{\alpha k}.
\end{align*}
For simplicity, let $A = \tp{\frac{2\Delta-1}{2\Delta}}^{1-\alpha}$, $B = \tp{\frac{1}{2\Delta}}^{1-\alpha} \cdot q^{\alpha}$. We have
\begin{align}\label{eq:beta-bound}
    \beta(\alpha) = \sum_{k=1}^{\Delta} (k+1) \binom{\Delta}{k} A^{\Delta-k} B^k = (A+B)^{\Delta - 1} (A+(\Delta+1)B) - A^\Delta
\end{align}
We first claim the following numeric bounds for $A$ and $B$:
\begin{align}\label{eq:numeric-bounds}
1 - \frac{1}{2\Delta} \le A \le 1 - \frac{0.99}{2\Delta} \text{ and } B \le \frac{1.01}{2\Delta}.
\end{align}
Plugging into~\eqref{eq:beta-bound}, we are able to show that $\beta(\alpha) \le 0.99$ for all $\Delta \ge 3$. 

It remains to show~\eqref{eq:numeric-bounds}. By a straightforward calculation, it holds
\begin{align*}
    B = \frac{(2q\Delta)^{\alpha}}{2\Delta} = \frac{\exp\tp{\alpha \log (2q\Delta)}}{2\Delta} \le \frac{1.01}{2\Delta}.
\end{align*}
Use inequality $1-x \le (1-x)^{1-\alpha} \le 1-(1-\alpha) x$ where $x, \alpha \in (0,1)$, we have
\begin{align*}
    1-\frac{1}{2\Delta}\le A \le 1 - \frac{1-\alpha}{2\Delta} \le 1-\frac{0.99}{2\Delta}.
\end{align*}
This concludes the proof.
\end{proof}

As a corollary, we have a sub-linear time algorithm to proceed $N=\Theta(n/\epsilon^2)$ perfect marginal samplers in a single batch for spin systems with soft constraints, thus implying a subquadratic time approximating algorithm for partition function.

\begin{corollary}[Soft-constraint spin system with weak interactions]\label{cor:spin}
    Let $\Delta \ge 3$ be a constant. There exist $\mathsf{FPRAS}$ for estimating the partition function of soft-constraint spin systems with sufficiently weak interaction on graphs with maximum degree $\Delta$ and $n$ vertices in expected time $\widetilde{O}(n^{2 - \frac{1}{1000 \log q\Delta}} \cdot \epsilon^{2 - \frac{1}{1000 \log q\Delta}})$, where $\epsilon > 0$ denotes the multiplicative error bound.
\end{corollary}

\subsection{Aggregate perfect marginal sampler for implicit perfect marginal sampler}\label{sec:implicit}
Unfortunately, not all perfect marginal samplers fall into the framework of~\Cref{alg:marginal-sampler-general} and~\Cref{thm:general-2}. 
For the perfect marginal sampler for hypergraph independent set, the expected number of visiting neighbors can exceed $1$, failing to satisfy the requisite in~\Cref{thm:general-explicit}.
For the perfect marginal sampler for proper colorings  in~\cite{LWY26}, it does not recursively reveal the colors of neighboring vertices. Instead, it asks for certificates whether a particular color is assigned to a neighbor vertex, which can not be expressed as~\Cref{alg:marginal-sampler-general}.

To resolve this challenge, we adopt the \emph{probabilistic automaton} framework. As a general computational model involving randomness, it naturally subsumes all perfect marginal samplers.
\begin{definition}[\cite{Rabin63}]
    A probabilistic automaton is a tuple $(\+I,\+S,i,\delta,\perp)$ where
    \begin{itemize}
    \item $\+I$ denotes the collection of inputs;
    \item $\+S$ denotes the collection of states and $i \in \+S$ denotes the initial state;
    \item $\delta : (S \setminus \perp) \times \+I \mapsto \+P(S)$ denotes the transition rules; Specifically, given input $I \in \+I$, the automaton transitions from state $x$ to other states according to the law of distribution $\delta_I(x,\cdot)$. 
    \item $\perp \subseteq \+S$ denotes the absorbing states reaching which the automaton would terminate and return the current state.
    \end{itemize}
\end{definition}

We now show that for probabilistic automaton with bounded width, efficient implementation and geometric tail, we can proceed $N$ independent runs jointly in sub-linear time.

\begin{theorem}\label{thm:batch-automaton}
        If the probabilistic automaton satisfies the following criteria:
        \begin{enumerate}
        \item (Bounded width) For each state $x \in \+S$ and input $I \in \+I$, the cardinality of the support of $\delta_I(x)$ is bounded by a constant $D \ge 2$;
        \item (Efficient implementation) For any transition path $x_0 = i, x_1,\ldots,x_\ell$ and input $I \in \+I$, the transition kernel $\delta_I(x_\ell)$ can be constructed in $O(\ell^c)$ time for some constant $c \ge 0$.
        \item (Exponential decay) Let $T$ be the stopping time of the probabilistic automaton. Specifically, let $x_0=i,x_1,\ldots,x_T$ be a transition path with $x_T \in \perp$. There exists $C > 0$ and $\alpha \in (0,1)$ such that for any input $I \in \+I$, it holds
        \begin{align*}
            \forall t \ge 0, \quad \Pr[]{T \ge t} \le C \cdot \alpha^{t}.
        \end{align*}
        \end{enumerate}
        Then, there is an algorithm simulating $N$ independent probabilistic automatons jointly, but only tracks the number of samples entering each absorbing state in $\widetilde{O}(N^{1-\eps})$ runtime, where $\eps = \frac{\log 1/\alpha}{\log D+\log 1/\alpha}$.
\end{theorem}

\begin{proof}
        Suppose we secretly run $N$ independent copies of the probabilistic automaton on input instance $I \in \+I$. Let $x^{(i)}_{0}, x^{(i)}_1,\ldots,x_{T_i}^{(i)}$ with $x_{T_i}^{(i)} \in \perp$ be the transition path of the $i$-th run.
        
         To proceed these $N$ independent runs jointly, we maintain sparse mappings $L_t = \{(x,M_{x,t})\}_{x \in \+S}$ for each time step $t \ge 0$. Here, $x$ denotes a state in $\+S$, and $M_{x,t}$ denotes the occurrences in $X^{(1)}_t, X_t^{(2)}, \ldots, X_t^{(N)}$ (If $X_t^{(i)}$ does not exist, we skip this item). We also maintain a sparse mapping $O = \{(x,\mathrm{occ})\}_{x \in \perp}$ to record the final frequencies of the absorbing states. For $t = 0$, initialize $L_0 = \{(i,N)\}$. Starting from $t = 0$, for each pair $(x,\mathrm{occ}) \in L_t$, we consider the following cases:
        \begin{enumerate}
        \item If $x \in \perp$ is one of the absorbing states, we update the entry $(x,z)$ in $O$ by $(x,z+\mathrm{occ})$;
        \item Otherwise, we construct the transition kernel $\delta_I(x_\ell,\cdot)$ and then draw from the multinomial distribution $M(\mathrm{occ};\delta_I(x_\ell))$. Suppose state $x_{\mathrm{next}}$ is picked $y$ times in the multinomial distribution. We update the entry $(x_{\mathrm{next}},M_{x_{\mathrm{next}},t+1})$ in $L_{t+1}$ by $(x_{\mathrm{next}},M_{x_{\mathrm{next}},t+1} + y)$.
        \end{enumerate}
        The process terminates when $L_{t}$ becomes empty.

        Let $N_\ell$ be the number of states with non-zero occurrences in $L_\ell$ and $P_\ell$ be the number of transition paths starting from $x_0 = i$ with non-zero probability. We notice that
        \begin{align*}
            \E[]{L_\ell} &\le \sum_{\substack{(x_0,x_1,\ldots,x_\ell) \in \+S^\ell\\x_0 = i}} \tp{1 -\tp{1- \prod_{i=1}^\ell \delta_{I}(x_{\ell-1},x_\ell)}^N} \\
            (\text{convexity of }x \mapsto x^N)\quad &{\le} P_\ell \cdot \tp{1 - \tp{1-\frac{1}{P_\ell} \sum_{\substack{(x_0,x_1,\ldots,x_\ell) \in S^\ell\\ x_0=i}} \prod_{i=1}^\ell \delta_I(x_{\ell-1},x_\ell)}^N}\\
            (\text{by bounded width, }P_\ell \le D^\ell)\quad &\le D^\ell \cdot \tp{1-\tp{1 - \frac{1}{D^\ell} \Pr[]{T \ge \ell}}^N}\\
            (\text{exponential decay})\quad &\le D^\ell \cdot \min\tp{1, CN \cdot \tp{\frac{\alpha}{D}}^\ell}.
        \end{align*}
        
        Together with~\Cref{lem:binomial}, the time complexity of algorithm is bounded by the following up to logarithmic factors:
        \begin{align*}
            \sum_{\ell \ge 0} \ell^c \cdot \E[]{L_\ell}  &\le \sum_{\ell \ge 0} \ell^c D^\ell \cdot \min\tp{1,C N \cdot \tp{\frac{\alpha}{D}}^\ell}\\
            (\text{set }B = \left\lceil \frac{\log (CN)}{\log (D/\alpha)}\right\rceil)\quad & \le \sum_{\ell=0}^B \ell^c D^\ell +\sum_{\ell=B}^{\infty} CN \cdot \ell^c \alpha^\ell\\
            &= \Theta(B^c D^B + CN \cdot B^c \alpha^B)\\
            &= \widetilde{\Theta}(N^{\frac{\log D}{\log D+\log 1/\alpha}}).
        \end{align*}
        This concludes the proof.
\end{proof}

\subsubsection{Applications}

In this section, we apply~\Cref{thm:batch-automaton} to perfect marginal samplers in various settings.
The high-level idea is to translate a recursive algorithm into a non-recursive one using \emph{call stack}, and then express it in the framework of probabilistic automaton. We demonstrate this transformation for the hardcore model in detail, and provide sketches for other models of interest.

\paragraph{Hardcore model}
We consider the seminal Anand-Jerrum marginal sampler (\Cref{alg:AJ-marginal-sampler-new}) as an illustrative example. Let $S$ denote the call stack maintaining the state of unresolved calls. Each element in $S$ takes the form $((u, \tau), \text{status})$, where the pair $(u, \tau)$ represents an unresolved function call with its associated parameters, and \texttt{status} indicates the current execution phase of that call.

The status of a function call $(G, u, \tau)$ is categorized into one of the following four phases:
\begin{description}
    \item[Initialization (\texttt{init}):] The entry point of an unresolved function call.
    \item[Inspection (\texttt{inspect}, $\tau$):] The stage where the algorithm is ready to inspect neighbor vertices.
    \item[Holding (\texttt{hold}, $\tau, u$):] A suspended state awaiting the return value from a recursive call.
    \item[Halting (\texttt{halt}, $0/1$):] The termination of a single function call with a return value of $0$ or $1$.
\end{description}
Initially, the stack $S = \{((u,\emptyset),\texttt{init})\}$. In each step, the process updates the stack $S$ according to its top (or top few) elements. Let $((w,\tau),\texttt{status})$ be the top element in the stack $S$, the process complies the following rules:
\begin{enumerate}
\item If \texttt{status} $=$ \texttt{init}, the process tosses a biased coin to either update the top element of stack $S$ by $((w,\tau), (\texttt{halt}, 0))$ with probability $\frac{1}{1+\lambda}$ or otherwise $((w,\tau),(\texttt{inspect}, \tau))$;
\item If \texttt{status} $=$ (\texttt{inspect}, $\tau^\star$) where $\tau^\star \in \{0,1\}^{\Lambda^\star}$, there are two cases to consider:
\begin{enumerate}
\item If $N(w) \subseteq \Lambda^\star$, update the top element of the stack $S$ by $((w,\tau), (\texttt{halt}, 1))$;
\item Otherwise, select $v \in N(w) \setminus \Lambda^\star$, update the top element of stack $S$ by $((w,\tau),(\texttt{hold}, \tau^\star, v))$, and then push $((w,\tau^\star),\texttt{init})$ into the stack.
\end{enumerate}
\item If \texttt{status} $=$ (\texttt{halt}, $\mathrm{sgn}$), where $\mathrm{sgn} \in \{0,1\}$, we first pop the top element from the stack $S$. There are two cases to consider:
\begin{enumerate}
    \item If $S$ is empty, the whole stochastic automaton reaches the termination state $\perp_{\mathrm{sgn}}$;
    \item Otherwise, the top element of the current stack $S$ has the form of $((w^\star,\tau^\star),(\texttt{hold}, \tau^\circ, v))$. 
    \begin{enumerate}
    \item If $\mathrm{sgn} = 1$, we update the top element by $((w^\star,\tau^\star),(\texttt{halt},0))$;
    \item Otherwise, replace by $((w^\star,\tau^\star), (\texttt{inspect},\tau^\circ \cup (v \gets 0 )))$.
    \end{enumerate}
\end{enumerate}
\end{enumerate}
We verify the prerequisites in~\Cref{thm:batch-automaton}.
\begin{enumerate}
    \item As it can be seen in the transition rules, the width is bounded by $2$;
    \item For any transition path $p = (i,x_1,x_2,\ldots,x_\ell)$, the transition step from $x_i$ to $x_{i+1}$ takes $O(i)$ time to implement;
    \item The termination time $T$ is at most a constant times of the number of recursive calls the whole process triggers. Specifically, $T$ is stochastically bounded by $3T_0$, where $T_0 \coloneqq \min \{t \ge 0 \mid \Phi_t = 0\}$ is the stopping time of the following $1$-dimensional random walk:
    \begin{align*}
        \Phi_0 = 1 \text{ and } \Phi_t - \Phi_{t-1} = 
        \begin{cases}
            -1 & \text{with probability $\frac{1}{1+\lambda}$,}\\
            \Delta-1 & \text{otherwise.}
        \end{cases}
    \end{align*}
    When $\lambda \le \frac{1-c}{\Delta -1}$ for some $c \in (0,1)$, by Chernoff bound it holds
    \begin{align*}
        \Pr[]{T_0 \ge \ell} \le \Pr[]{\Phi_\ell \ge 1} \le \exp\tp{-\frac{\ell}{2(1+\lambda)}\tp{\frac{1-(\Delta-1) \lambda}{\Delta}}^2} \le \exp\tp{-\frac{c^2}{4\Delta^2} \cdot \ell}.
    \end{align*}
\end{enumerate}
Plugging $\alpha = \exp\tp{-\frac{c^2}{4\Delta^2}}$ into~\Cref{thm:batch-automaton}, the runtime for simulating $N$ independent runs jointly is $\widetilde{O}(N^{1-\eps})$ where $\eps = \frac{c^2}{4\Delta^2 + c^2}$.

\paragraph{Hypergraph independent sets}

The perfect marginal sampler for hypergraph independent sets in~\cite{FGWWY23} is based on the ``coupling towards the past'' paradigm. At a high level, this is an algorithmic realization of the information percolation technique: one unrolls the randomness backward in time and tracks the (random) dependency cluster of a queried site. Under suitable conditions, this dependency cluster is typically small, leading to efficient termination while preserving the correct marginal distribution.

We first recall the perfect marginal sampler for hypergraph independent sets in~\cite{FGWWY23}.
Let $G=([n],E)$ be a hypergraph, and let $\mu$ denote the uniform distribution over independent sets.
Consider the systematic-scan Glauber dynamics $(X_t)_{t\ge 0}$ starting from an arbitrary initial state $X_0$.
In round $t$, the chain updates the vertex $v_t \coloneqq t \bmod n$ as follows.
First draw a coin $r_t \sim \mathrm{Ber}(1/2)$.
If $r_t=0$, then the update sets $X_t(v_t)\gets 0$.
If $r_t=1$, the update attempts to set $X_t(v_t)\gets 1$, but this is vetoed whenever doing so would violate feasibility:
namely, if there exists a hyperedge $e\in E$ with $v_t\in e$ such that all other vertices in $e\setminus\{v_t\}$ are currently occupied, then we must set $X_t(v_t)\gets 0$; otherwise we set $X_t(v_t)\gets 1$.

The perfect marginal sampler can be viewed as the time-reversal exploration process of this systematic scan.
Suppose we would like to determine the value of $X_t(v_t)$ at some time $t\le 0$ in the stationary limit (formally, as the starting time tends to $-\infty$).
We first reveal the random bit $r_t$.
If $r_t=0$, then $X_t(v_t)=0$ immediately.
If $r_t=1$, then $X_t(v_t)$ equals $1$ unless some incident hyperedge $e\ni v_t$ forces it to be unoccupied.
To check whether an edge $e$ forces $v_t$, we must know whether all vertices $u\in e\setminus\{v_t\}$ were last updated to $1$ before time $t$.
This induces a recursive dependency on earlier times: for each such $u$, we look at $\mathrm{prev}(u,t)$, the most recent time $t'\le t$ at which vertex $u$ is updated, and recursively resolve $X_{\mathrm{prev}(u,t)}(u)$.
Taking the limit where the time horizon goes to $-\infty$, this procedure yields a perfect sample from the stationary marginal distribution.

We now give a formal description.

\begin{algorithm}[H]
    \KwIn{A hypergraph $G=(V,E)$ and a time $t\le 0$}
    \KwOut{The spin $X_t(v_t)\in\{0,1\}$, where $v_t \coloneqq t \bmod n$}
    \KwGlobal{Random bits $(r_s)_{-\infty < s \le 0}$ and memo table $(Y_s)_{-\infty < s \le 0}$ initialized with $Y_s=\perp$}
    \caption{Marginal sampler for hypergraph independent set: $\mathrm{Resolve}(G,t)$}
    \label{alg:hypergraph-IS-marginal}
    \If{$Y_t \neq \perp$}{
        \Return $Y_t$\;
    }
    \If{$r_t = 0$}{
        \Return $Y_t \gets 0$\;
    }
    \ForEach{$e \in E$ with $v_t \in e$}{
        \If{$\forall u \in e \setminus \{v_t\}$, $r_{\mathrm{prev}(u,t)} = 1$}{
            \If{$\forall u \in e \setminus \{v_t\}$, $\mathrm{Resolve}(G,\mathrm{prev}(u,t)) = 1$}{
                $Y_t \gets 0$\;
                \Return $0$\;
            }
        }
    }
    \Return $Y_t \gets 1$\;
\end{algorithm}

Here $\mathrm{prev}(u,t)$ denotes the maximum integer $t'\le t$ such that $t' \equiv u \pmod n$, i.e., the last update time of vertex $u$ before time $t$.

This procedure can be cast in the probabilistic automaton framework by taking the automaton state to consist of (i) the current call stack of $\mathrm{Resolve}$, together with (ii) the set of revealed global variables $\{r_s\}_{s\le 0}$ and memoized outputs $\{Y_s\}_{s\le 0}$.
The only source of branching is the first revelation of an unrevealed coin $r_s\in\{0,1\}$; hence every configuration has at most two outgoing probabilistic transitions, and the resulting automaton has width $2$. And obviously, each transition step can be efficiently implemented, and most importantly, the process exhibits a stopping time with geometric decay, which is suggested by the following proposition.

\begin{proposition}[Lemma 5.9, \cite{FGWWY23}]
    Let $R$ be the number of random bits accessed during the process. Upon the termination of the perfect marginal sampler, it holds
    \begin{align*}
        \Pr[]{R \ge 3 \Delta^2 k^4 \cdot \eta} \le \exp\tp{-\eta}. 
    \end{align*}
\end{proposition}

Therefore, we are able to apply~\Cref{thm:batch-automaton} to obtain the sub-linear time aggregate sampler and thus a subquadratic time algorithm for estimating the partition function.

\begin{corollary}[Hypergraph independent set]\label{cor:hyper-indset}
    Let $\Delta \ge 3$ and $k \ge 2$ be constants satisfying $2^{k/2} \ge \sqrt{8\e} k^2 \Delta$. There exist $\mathsf{FPRAS}$ for estimating the number of independent sets in $k$-uniform hypergraph with maximum degree $\Delta$ and $n$ vertices in expected time $\widetilde{O}(n^{2-\delta} \eps^{-2+2\delta})$, where $\delta = \Omega_{\Delta,k}(1)$ and $\epsilon$ denotes the multiplicative error bound.
\end{corollary}

\paragraph{Proper colorings}
    The perfect marginal sampler we consider for the proper colorings is proposed in~\cite{LWY26}. At high level, their perfect marginal sampler is a refined implementation of the ``coupling towards the past'' framework. 
    
    Let $G=(V,E,L)$ be a graph together with lists of colorings $L=(L_v)_{v \in V}$. Consider the systematic-scan Glauber dynamics $(X_t)_{t \ge 0}$ starting from an arbitrary initial state $X_0$. In the round $t$, the Glauber dynamics would update the vertex $v_t \coloneqq t \mod n$ by one of the available colors $\+A_{t} = \{c \in L_{t} \mid \forall u \in N(v_t),X_{t-1}(u) \neq c\}$ uniformly at random, where $L_t \coloneqq L_{v_t}$.

    To implement the time-reversal exploration process efficiently, Liu, Wang and Yin observed that instead of querying the exact colors of the neighboring vertices, we only needs to certify if a particular color $c$ is assigned to $X_t(v_t)$ at time $t$ or not.
    Specifically, their algorithm includes a subroutine $\texttt{Certify}(t,c)$ which outputs a one-bit information about whether $X_t(v_t) = c$ or not. With this subroutine, we are able to resolve the color $X_t(v_t)$ by repeatedly proposing candidates $c \in L_{v_t}$ until $c \in \+A_t$. Now it remains to implement the \texttt{Certify} procedure. There are two separate cases to consider:
    \begin{enumerate}
    \item If $\abs{L_{t}} \le 50\Delta$, call the resolve procedure to determine the exact color of $X_t(v_t)$;
    \item Otherwise, the process should output $1$ with probability $\frac{1}{\abs{\+A_t}}$ if $c \in \+A_t$, and deterministically output $0$ if $c \not\in \+A_t$. This is effectively the same by immediately returning $0$ (and excludes color $c$ from $L_t$) with probability $1-\frac{2}{\abs{L_t}}$, and then check if $c \in \+A_t$ (by recursively calling the \texttt{Certify} procedure). When positive, the process returns $1$ with probability $\frac{\abs{L_t}}{2\abs{\+A_t}}$ via Bernoulli factory. In all other scenarios, the process returns $0$ (and excludes color $c$ from $L_t$).
    \end{enumerate}
    
    The formal description of this algorithm can be found in Algorithm 7-11 of~\cite{LWY26}. We would like to remark that this procedure can also be expressed in the probabilistic automaton framework by taking the automaton state to consist of (i) current call stacks of the \texttt{Resolve} and \texttt{Certify} process, together with (ii) the current available colors $(L_t)_{t \ge 0}$. The branching comes from proposing a uniformly random colors in $L_t$, or toss a random (biased) coin; hence the resulting automaton has width at most $q$. Each transition step can be implemented efficiently, and most importantly, the automaton also exhibits a stopping time with geometric decay, suggested by the following proposition.

    \begin{proposition}[\cite{LWY26}]
        Let $q,\Delta$ be constants with $q \ge 65 \Delta$.
        Let $\tau$ be the total number of function calls when resolving the color at time $t$. It holds for any $\eta \ge 60\Delta$ that
        \begin{align*}
            \Pr[]{\tau \ge \eta} \le \exp\tp{-\Omega_{q,\Delta}\tp{\eta}}.
        \end{align*}
    \end{proposition}

    Therefore, we are able to apply~\Cref{thm:batch-automaton} to obtain the subquadratic time algorithm for estimating the partition function.

    \begin{corollary}[Proper colorings]\label{cor:coloring}
        Let $\Delta \ge 3$ and $q$ be constants satisfying $q \ge 65 \Delta$. There exists $\mathsf{FPRAS}$ for estimating the number of proper colorings in expected time $\widetilde{O}(n^{2-\delta} \eps^{-2+2\delta})$, where $\delta = \Omega_{\Delta,q}(1)$ and $\eps$ denotes the multiplicative error bound.
    \end{corollary}

\ifthenelse{\boolean{anonymous}}{}{
\section*{Acknowledgment}
We would like to thank Hongyang Liu, Chunyang Wang and Yitong Yin for insightful discussions.
}
\bibliographystyle{alpha}
\bibliography{refs}

\appendix
\section{Missing proofs}
\subsection{Proof of \texorpdfstring{\Cref{lem:reduction-counting-to-inference}}{Lemma 2.5}}

\begin{proof}[Proof of \Cref{lem:reduction-counting-to-inference}]
    Fix $\sigma \in [q]^n$ such that $\mu(\sigma) > 0$.
    Let $\sigma(< i) \coloneqq (\sigma_1, \cdots, \sigma_i)$ for $i \in [n]$.
    According to the chain rule, we have
    \begin{align*}
        \mu(\sigma) = \prod_{i=1}^n \mu_i^{\sigma(<i)}(\sigma_i).
    \end{align*}
    Let $k$ be an integer (universal constant) that we will determine later.
    For each $i$, we run $\mathcal{A}$ for $k$ times with time $k \cdot T(n)$ to compute $P_i \coloneqq \frac{1}{k} \sum_{j=1}^k P_{i,j}$ independently such that $\E{P_{i,j}} = \mu_i^{\sigma(<i)}(\sigma_i)$ and $\Var{P_{i,j}}/\E{P_{i,j}}^2 \leq \epsilon^2/n$.
    This implies that $\E{P_i} = \mu_i^{\sigma(<i)}(\sigma_i)$ and $\Var{P_{i,j}}/\E{P_{i,j}}^2 \leq \epsilon^2/(nk)$.
    Then, let $M \coloneqq \prod_{i=1}^n P_i$, we have $\E{M} = \mu(\sigma)$ and 
    \begin{align*}
        \frac{\Var{M}}{\E{M}^2} 
        &= \frac{\E{M^2}}{\E{M}^2} - 1
        = \prod_{i=1}^n \frac{\E{P_i^2}}{\E{P_i}^2} - 1 \tag{by independence} \\
        &= \prod_{i=1}^n \tp{1 + \frac{\Var{P_i}}{\E{P_i}^2}} - 1
        \leq \tp{1 + \frac{\epsilon^2}{k n}}^n - 1 \leq \e^{\epsilon^2/k} - 1\leq \frac{1}{k/\epsilon^2-1},
    \end{align*}
    where in the last inequality, we use the fact that $\e^{1/x} \leq 1 + 1/(x-1)$ for $x > 1$.
    We pick $k = 13$. 
    By Chebyshev's inequality, we have
    \begin{align*}
        \Pr{\abs{M - \E{M}} \geq \frac{\epsilon}{2} \E{M}} \leq \frac{4 \Var{M}}{\epsilon^2 \E{M}^2} \leq \frac{4}{\epsilon^2 (13/\epsilon^2-1)} \leq 1/3.
    \end{align*}
    Since $\e^{-\epsilon} \leq 1-\epsilon/2 \leq 1 + \epsilon/2 \leq \e^{\epsilon}$ for $\epsilon \in (0,1)$, it implies that 
    \begin{align*}
        \Pr{\e^{-\epsilon} \mu(\sigma) \leq M \leq \e^\epsilon \mu(\sigma)} &\geq \frac{3}{4}. \qedhere
    \end{align*}
\end{proof}

\subsection{Proof of \texorpdfstring{\Cref{lem:symmetrization}}{Lemma 3.16}} \label{sec:symmetrization}

\begin{proof}[Proof of \Cref{lem:symmetrization}]
  Following the proof of \cite[Lemma 13]{li2013correlation}, we define $z_i \coloneqq \frac{\beta x_i + 1}{x_i + \gamma}$.
  This means $z_i \in (\beta, 1/\gamma]$ and $x_i = \frac{1-\gamma z_i}{z_i - \beta}$.
  Hence, it holds that
  \begin{align*}
    h(x_i) = \frac{(z_i^{-1} - \gamma)(z_i - \beta)}{1-\beta\gamma}.
  \end{align*}
  Define $\ol{z} \coloneqq (\prod_{i=1}^d z_i)^{1/d}$ and $\ol{x} \coloneqq \frac{1 - \gamma \ol{z}}{\ol{z} - \beta}$.
  Hence it is direct to verify that
  \begin{align*}
    F_d(\*1\ol{x}) = \lambda \tp{\frac{\beta \ol{x} + 1}{\ol{x} + \gamma}}^d = \lambda \ol{z}^d = \lambda \prod_{i=1}^d \frac{\beta x_i + 1}{x_i + \gamma} = F_d(\*x).
  \end{align*}
  Define the following function
  \begin{align*}
    g(y) \coloneqq (\e^{-y} - \gamma)(\e^y - \beta).
  \end{align*}
  It is direct to show that $g(y)$ is concave on $\mathbb{R}$ by noticing that
  \begin{align*}
    g''(y) = -\e^{-y} \left(\beta +\gamma  \e^{2 y}\right) < 0.
  \end{align*}
  Let $y_i \coloneqq \log z_i$ and $y = \log\ol{z}$, we can rewrite $\sum_{i=1}^dh(x_i) \leq d h(\ol{x})$ as
  \begin{align*}
    \frac{1}{d}\sum_{i=1}^d g(y_i) \leq g\tp{\frac{1}{d} \sum_{i=1}^d y_i},
  \end{align*}
  which follows by Jensen's inequality.
\end{proof}

\subsection{Algorithmic self-avoiding walk tree with flower} \label{sec:SAW-tree-with-flower}
In this section, we will prove \Cref{lem:SAW-tree} and \Cref{lem:SAW-tree-with-flower}.
We first recall the settings.
Fix $\beta, \gamma, \lambda$ such that $0 \leq \beta \leq \gamma$ and $\gamma, \lambda > 0$.
Let $G = (V, E)$ be a graph with max degree $\Delta$.
Let $<$ be a total order in $V$.
Let $\Lambda \subseteq V$ and $\tau \in \Omega(\mu_{G,\Lambda})$.

\begin{definition}\label{def:SAW-one-step}
  Let $G = (V,E)$ be a graph.
  Let $\tau$ be a pinning on $\Lambda \subseteq V$.
  Let $(v,u)$ be a pair of neighoring vertices in $G$.
  Let $\Gamma^G_v = \set{u_1, \cdots, u_d}$ be the set of neighbors of $v$ in $G$ sorted ascendingly according to $<$.
  If $u$ is not pinned by $\tau$, we define a new graph $G^{(v,u)}$ (with pinning $\tau^{(v,u)}$) as follows:
  \begin{itemize}
  \item Remove $v$ and for each $i \in [d]$ such that $u_i \neq u$:
    \begin{enumerate}
    \item add a copy $v_i$ of $v$ that is pinned to $\tau^{(v,u)}_{v_i} \coloneqq \*1[u_i > u]$;
    \item add a new edge $\set{v_i, u_i}$;
    \end{enumerate}
  \item For vertex $s \in V\setminus\set{v}$, their pinning are inherited from $\tau$ (i.e., we set $\tau^{(v,u)}_s \coloneqq \tau_s$).
  \end{itemize}
  Otherwise, if $u$ is pinned by $\tau$, we set $G^{(v,u)} = \set{u}$ and $\tau^{(v,u)}_u = \tau_u$.
\end{definition}

\begin{definition}[self-avoiding walk] \label{def:SAW}
  Let $G = (V,E)$ be a graph.
  Let $\tau$ be a pinning on $\Lambda \subseteq V$.
  we define self-avoiding walks $w = (v_0, v_1, \cdots, v_\ell)$ with $\ell \geq 1$ and corresponding graph $G^w$ with pinning $\tau^w$ inductively as follows:
  \begin{itemize}
  \item For $\ell = 0$, we say $w = (v_0)$ is a self-avoiding walk if $v_0$ is an vertex in $G$. Then we set $G^w = G$ and $\tau^w = \tau$;
  \item For $\ell > 0$, $w = (v_0, \cdots, v_\ell)$ is a self-avoiding walk if its prefix $w' = (v_0, \cdots v_{\ell-1})$ is a self-avoiding walk and $\set{v_{\ell-1}, v_\ell}$ is an edge in $G^{w'}$.
    Then construct the graph $G^w \coloneqq (G^{w'})^{(v_{\ell-1}, v_\ell)}$ and pinning $\tau^w \coloneqq (\tau^{w'})^{(v_{\ell-1}, v_\ell)}$ as in \Cref{def:SAW-one-step}.
  \end{itemize}
\end{definition}

\begin{definition}[SAW tree] \label{def:SAW-tree}
  Let $G = (V,E)$ be a graph.
  Let $\tau$ be a pinning on $\Lambda \subseteq V$.
  Let $v$ be a vertex in $G$.
  We define the SAW tree $\mathbb{T}_{\mathrm{SAW}(v)}(G) = (V_{\mathrm{SAW}(v)}, E_{\mathrm{SAW}(v)})$ as the dictionary tree of all the self-avoiding walks $w = (v_0, \cdots )$ with $v_0 = v$.
  In particular, we have
  \begin{itemize}
  \item $V_{\mathrm{SAW}(v)} \coloneqq \set{\text{self-avoiding walk } w = (v_0, \cdots) \mid v_0 = v}$;
  \item for $w, w' \in V_{\mathrm{SAW}(v)}$, $w$ is a child of $w'$ if and only if $w$ extend $w'$ by exactly one more vertex.
  \end{itemize}
  For every $w = (v_0, \cdots, v_\ell) \in V_{\mathrm{SAW}(v)}$, we define its pinning as $\tau_{\mathrm{SAW}(v), w} \coloneqq \tau^w_{v_\ell}$.
\end{definition}

\begin{remark}\label{rem:SAW-tree}
  Let $\mathbb{T}_{\mathrm{SAW}(v)}(G)$.
  Let $w = (v_0,\cdots, v_\ell)$ be a node in $\mathbb{T}_{\mathrm{SAW}(v)}(G)$.
  We will use $\mathbb{T}_w$ to denote the sub-tree rooted at $w$.
  By \Cref{def:SAW-tree}, it is direct to see that $\mathbb{T}_w = \mathbb{T}_{\mathrm{SAW}(v_\ell)}(G^w)$.
  Moreover, we have $\tau_{\mathrm{SAW}(v)}\vert_{\mathbb{T}_w} = \tau^w_{\mathrm{SAW}(v_\ell)}$.
\end{remark}

\begin{proof}[Proof of \Cref{lem:SAW-tree}]
  We first prove $\mu^\tau_{G,v} = \mu^{\tau_{\mathrm{SAW}(v)}}_{\mathbb{T}_{\mathrm{SAW}(v)}(G), v}$ by doing induction on the size of SAW tree.
  Without loss of generality, we will assume that $v$ is free under $\tau$.
  The base case is trivial, since $|V_{\mathrm{SAW}(v)}| = 1$ implies that $v$ is an isolated vertex in $G$.
  For the inductive case, let $\Gamma^G_v = \set{u_1, \cdots, u_d}$ be the set of neighbors of $v$ in $G$ ordered by $<$ ascendingly.
  For each $i \in [d]$, we define $w_i \coloneqq (v, u_i)$ and let $\mathbb{T}_{w_i}$ be the sub-tree of SAW-tree rooted at $w_i$.
  Then by \Cref{rem:SAW-tree}, we have $\mathbb{T}_{w_i} = \mathbb{T}_{\mathrm{SAW}(u_i)}(G^{w_i})$ and $\tau_{\mathrm{SAW}(v)}\vert_{\mathbb{T}_{w_i}} = \tau^w_{\mathrm{SAW}(u_i)}$
  By the tree recursion, we have
  \begin{align*}
    R^{\tau_{\mathrm{SAW}(v)}}_{\mathbb{T}_{\mathrm{SAW}(v)}(G), v}
    &= \lambda \prod_{i=1}^d \frac{\beta R^{\tau_{\mathrm{SAW}(v)}\vert_{\mathbb{T}_{w_i}}}_{\mathbb{T}_{w_i}, w_i} + 1}{R^{\tau_{\mathrm{SAW}(v)}\vert_{\mathbb{T}_{w_i}}}_{\mathbb{T}_{w_i}, w_i} + \gamma}
    = \lambda \prod_{i=1}^d \frac{\beta R_{\mathbb{T}_{\mathrm{SAW}(u_i)}(G^{w_i}), u_i}^{\tau_{\mathrm{SAW}(u_i)}^{w_i}} + 1}{R_{\mathbb{T}_{\mathrm{SAW}(u_i)}(G^{w_i}), u_i}^{\tau_{\mathrm{SAW}(u_i)}^{w_i}} + \gamma}
   = \lambda \prod_{i=1}^d \frac{\beta R^{\tau^{w_i}}_{G^{w_i}, u_i} + 1}{R^{\tau^{w_i}}_{G^{w_i}, u_i} + \gamma},
  \end{align*}
  where the last equation follows from the induction hypothesis.
  Now define a new graph $\hat{G}$ from $G$ by spliting $v$ into $d$ copies where each copy $v_i$ is only connected to $u_i$.
  For $0 \leq i \leq d$, let $\sigma^{(i)}$ be pinings on $\set{v_1, \cdots, v_d}$ such that $\sigma^{(i)}_{v_j} = \*1[j > i]$.
  Then, by the tree recursion and \Cref{def:SAW-one-step},
  \begin{align*}
    \frac{\beta R^{\tau^{w_i}}_{G^{w_i}, u_i} + 1}{R^{\tau^{w_i}}_{G^{w_i}, u_i} + \gamma} &= \frac{1}{\lambda} \frac{\mu^\tau_{\hat{G}}(\sigma^{(i-1)})}{\mu^\tau_{\hat{G}}(\sigma^{(i)})}.
  \end{align*}
  Hence, by a telescoping product, we have
  \begin{align*}
    R^{\tau_{\mathrm{SAW}(v)}}_{\mathbb{T}_{\mathrm{SAW}(v)}(G), v}
    &= \frac{\lambda}{\lambda^d} \frac{\mu^{\tau}_{\hat{G}}(\sigma^{(0)})}{\mu^{\tau}_{\hat{G}}(\sigma^{(d)})} = R^\tau_{G,v}.
  \end{align*}
  This finishes the first bullet of \Cref{lem:SAW-tree}.

  For the second bullet, we can maintain $w = (v_0, \cdots, v_\ell), G^w, \tau^w$ in the data structure $\mathcal{D}$.
  The children of $w$ can be enumerated in $O_\Delta(1)$ time by looking at neighbors of $v_\ell$ in $G^w$.
  The parent of $w$ is direct from the $\ell-1$ prefix of $w$.
  When an update comes, by \Cref{def:SAW-one-step}, we can move to children of $w$ in $O_\Delta(1)$ time.
  To move to the parent $w'$ of $w$, we simply rebuild $G^{w'}$ and $\tau^{w'}$ from scratch.
  This would cost $O_\Delta(\ell)$ time.
  This finishes the proof.
\end{proof}

\begin{proof}[Proof of \Cref{lem:SAW-tree-with-flower}]
Given a set $S\subseteq V_{\mathrm{SAW}(v)}$, we construct the SAW tree with flower $G_{\mathrm{SF}}(G,v,S)$ by, for every $w \in S$, replacing the sub-tree $\mathbb{T}_{w}$ rooted at $w$ by $G^w$ with pinning $\tau^w$.
By \Cref{rem:SAW-tree} and the first bullet of \Cref{lem:SAW-tree}, it is direct to see that the marginal distribution of the root is preserved.
This proves the first bullet of \Cref{lem:SAW-tree-with-flower}.

  For the second claim, we define 
  \begin{align*}
    \mathcal{T} \coloneqq G_{\mathrm{SF}}[S \cup \set{\text{the connected component in $G_{\mathrm{SF}} - S$ that contains root $v$}}]
  \end{align*}
  We first scan $\mathcal{T}$ and given an index to every vertex in $G_{\mathrm{SF}}$.
  This index will help us figure out whether a given vertex $u \in V_{\mathrm{SF}}$ is a self-avoiding walk in $V_{\mathrm{SAW}}$ or is a copied vertex in $G^w$ for some $w \in S$.
  During this scan of $\mathcal{T}$, we also build the adjacency structure of $\mathcal{T}$ and the pinning for every vertex $w \in \mathcal{T}$.
  This initialization scan can be achieved by using the data structure $\mathcal{D}$ in \Cref{lem:SAW-tree} and will take time $\abs{\mathcal{T}} \cdot T_{\mathcal{D}} = O_\Delta(\dep_{\mathbb{T}_{\mathrm{SAW}}}(S) \cdot \siz_{\mathbb{T}_{\mathrm{SAW}}}(S))$.

  After the initialization, we can answer both queries supported by $\mathcal{O}$ for $u \in V_{\mathrm{SF}}$ as follow:
  \begin{itemize}
  \item If $u \in \mathcal{T}$, then return the answer that is already stored in memory in $O_\Delta(1)$ time.
  \item If $u \in G^w$ for some $w \in S$, then we build $G^w$ from $G$ without making copy. This takes at most $O_\Delta(\dep_{\mathbb{T}_{\mathrm{SAW}}}(S))$ time. After we have $G^w$, we can answer the question in $O_\Delta(1)$ time.
    After we answer the query, we undo the modification on $G^w$ and recover graph $G$.
  \end{itemize}
  This finishes the proof.
\end{proof}

\end{document}